\documentclass[12pt,a4paper]{article}
\usepackage{amsmath,amsthm,amsfonts,amssymb,bbm,color}
\usepackage{graphicx,psfrag,subfigure,url}
\usepackage{cite}
\usepackage{hyperref}
\usepackage[english]{babel}

\newcommand{\Id}{\mathbbm{1}}
\newcommand{\Or}{\mathcal{O}}
\newcommand{\E}{\mathbbm{E}}
\newcommand{\e}{\varepsilon}
\newcommand{\I}{{\rm i}}
\newcommand{\C}{\mathbb{C}}
\newcommand{\R}{\mathbb{R}}
\newcommand{\N}{\mathbb{N}}
\newcommand{\Z}{\mathbb{Z}}
\newcommand{\dx}{\mathrm{d}}
\newcommand{\QHyp}[6]{{}_{#1}\phi_{#2}\!\left(\begin{matrix}#3\\#4\end{matrix} \, \bigg| \, #5;#6\right)}

\renewcommand{\Re}{\operatorname{Re}}

\numberwithin{equation}{section}

\DeclareMathOperator*{\Tr}{Tr}
\DeclareMathOperator*{\Var}{Var}

\DeclareMathOperator{\Ai}{Ai}

\DeclareMathOperator{\Pb}{\mathbbm{P}}

\DeclareMathAlphabet{\mathpzc}{OT1}{pzc}{m}{it}

\newtheorem{result}{Result}
\newtheorem{prop}{Proposition}[section]

\newtheorem{lem}[prop]{Lemma}

\newtheorem{cor}[prop]{Corollary}

\newtheorem{rem}[prop]{Remark}
\newenvironment{remark}{\begin{rem}\normalfont}{\end{rem}}

\title{Finite time corrections in KPZ growth models}

\author{Patrik L. Ferrari\thanks{Institute for Applied Mathematics, University of Bonn, Endenicher Allee 60,\newline 53115 Bonn, Germany; E-mail:~\texttt{ferrari@uni-bonn.de}} ,
Ren\'e Frings\thanks{Institute for Applied Mathematics, University of Bonn, Endenicher Allee 60,\newline 53115 Bonn, Germany; E-mail:~\texttt{frings@uni-bonn.de}}}

\date{8. August 2011}

\begin{document}
\sloppy
\maketitle

\begin{abstract}
We consider some models in the Kardar-Parisi-Zhang universality class, namely the polynuclear growth model and the totally/partially asymmetric simple exclusion process. For these models, in the limit of large time $t$, universality of fluctuations has been previously obtained. In this paper we consider the convergence to the limiting distributions and determine the (non-universal) first order corrections, which turn out to be a non-random shift of order $t^{-1/3}$ (of order $1$ in microscopic units). Subtracting this deterministic correction, the convergence is then of order $t^{-2/3}$. We also determine the strength of asymmetry in the exclusion process for which the shift is zero. Finally, we discuss to what extend the discreteness of the model has an effect on the fitting functions.
\end{abstract}

\section{Introduction and results}\label{SectIntro}
A growth model is a stochastic evolution for a height function $h_t(x)$, $x$ space, $t$ time. In this paper we consider some models in the Kardar-Parisi-Zhang (KPZ) universality class~\cite{KPZ86} in $1+1$ dimensions, that are irreversible and have local growth rules. Moreover, there is a smoothing mechanism ensuring the existence of a deterministic limit shape \mbox{$h_{\rm ma}(\xi):=\lim_{t\to\infty} t^{-1} h_t(\xi t)$}, around which fluctuations are expected to have some degree of universality. More precisely, the height function under the scaling
\begin{equation}\label{eqIntro1}
h^{\rm resc}_t(u):=\frac{h_t(\xi t+u t^{2/3})-t\, h_{\rm ma}(\xi+u t^{-1/3})}{t^{1/3}}
\end{equation}
should be a well-defined stochastic process as $t\to\infty$, which means that fluctuations are of order $t^{1/3}$ and the spatial correlation length scales as $t^{2/3}$ (below we will focus on the one-point distribution only). The limit process still depends on initial conditions and is expected to be universal inside subclasses of the KPZ class (e.g., the processes for flat or curved limit shape are different). For more details, see the recent reviews~\cite{Fer10b,FS10} and the references therein.

\medskip
\emph{Theory.} On the theoretical side, some solvable models in the KPZ class have been analyzed in great detail. Two of the best studied models are the polynuclear growth (PNG) model and the (totally/partially) asymmetric simple exclusion process (TASEP/PASEP), see below for a concise definition. For example, it was shown~\cite{Jo00b,PS00,BR00,TW08b} that if the initial conditions generate a curved limit shape, then the limiting distribution function of $F_t=\Pb(h^{\rm resc}_t(0)\leq s)$ as $t\to\infty$ is given by the GUE Tracy-Widom distribution function, $F_{\rm GUE}$, occurring first in random matrix theory in~\cite{TW94}. Recently, the first model in the KPZ class, namely the KPZ equation itself, it was solved and the $F_{\rm GUE}$ was obtained~\cite{SS10,SS10a,SS10b,ACQ10}\footnote{These works uses the approach of Bertini and Giacomin for the weakly asymmetric simple exclusion process (WASEP)~\cite{BG97}; a replica approach is in~\cite{CDR10,Dot10,SP11}; see the review~\cite{SS10c} for details.}.

This solution of the KPZ equation includes the distribution function for $h_t$ at any time $t$, not only in the $t\to\infty$ limit. A further aspect noticed in~\cite{SS10b} is that the correction of $F_t$ from $F_{\rm GUE}$ is of order $t^{-1/3}$. This means that on the original scale, the difference between the height function $h_t(\xi t)$ and $t\, h_{\rm ma}(\xi)$ is of order $1$.

\medskip
\emph{Experiments.} Until recently, on the experimental side there were only few experiments giving the fluctuations exponent $1/3$~\cite{WIMM97,MMK97}. Besides the difficulties of having good statistics, one of the main issues in the experimental set-up is to have \emph{really a local} dynamic, and the centering in~(\ref{eqIntro1}) has to be obtained experimentally from the measured asymptotic growth velocity. In any case, experimental data were not good enough to have more detailed information on the scaling exponents, until the recent amazing experiments carried out by Takeuchi (see~\cite{TS10} and~\cite{TSSS11}). Using nematic liquid crystals they were able to get accurate statistics that confirmed the fluctuation and correlation exponents, but also the limiting distribution functions and the covariance of the processes\footnote{The processes are: (a) for non-random initial condition and flat limit shape: the Airy$_1$ process~\cite{Sas05,BFPS06}, (b) for curved limit shape: the Airy$_2$~\cite{PS02} process.} previously obtained in solvable models.

As it was the case for the solution of the KPZ equation cited above, also in these experiments one could see that the fit between the density of the Tracy-Widom distributions and the measurements is good even at relatively small time $t$, but finite size corrections are still visible. They measured also the decay of the mean, variance, skewness and kurtosis. In the scaled variables, the mean has been seen to decay as $t^{-1/3}$, while the others as $t^{-2/3}$. Thus, in the unrescaled variables, the mean has a shift of order $1$.

\medskip
Before explaining the results, we briefly describe the models analyzed in this paper\footnote{Java animations of these models can be found for the PNG model at \url{http://www-wt.iam.uni-bonn.de/~ferrari/animations/AnimationRSK.html}  and at
\url{http://www-wt.iam.uni-bonn.de/~ferrari/animations/ContinuousTASEP.html} for TASEP.}.

\emph{(a) PASEP and TASEP.} The partially asymmetric simple exclusion process (PASEP) on $\Z$ in continuous time is an interacting particles system. At any time instant $t$ at most one particle can occupy a site in $\Z$ and particles try to jump to the right neighboring site with rate $p$ and to the left neighboring site with rate $q=1-p$. The jumps are made only if the arrival sites are free. This dynamics does not change the order of particles. We label them from right to left so that $x_k(t)$ denotes the position of particle with label $k$ at time $t$ and $x_k(t)>x_{k+1}(t)$ for all $k$ and $t$. The two initial conditions analyzed in this paper are \emph{step initial condition}, $x_k(0)=-k$, $k=1,2,\ldots$, and \emph{alternating initial condition}, $x_k(0)=-2k$, $k\in\Z$.
When $q=0$ (and $p=1$) the particles can only jump to the right so that the model is totally asymmetric and is called TASEP.

\emph{(b) PNG model.} The polynuclear growth model describes an interface given by a height function $x\mapsto h_t(x)\in\Z$, $x\in \R$, which is a step function with up- and down-steps on the integers and constant in between. The PNG dynamics is the following. The up-steps move to the left with unit speed and the down-steps to the right with unit speed. On top of it, there are random nucleation events, i.e., creation of a pair of up- and down-steps, which then follow the deterministic spread to the left/right. In this paper we consider $h_0(x)=0$ for all $x\in\R$ and nucleations as follows: for \emph{flat PNG}, nucleations form a Poisson point process in $\R\times \R_+$ with intensity $2$, while for \emph{PNG droplet} the nucleations are further restricted to $\{(x,t), |x|\leq t\}$ (see also the review~\cite{FP05} for more details and illustrations).

\medskip
\emph{Results.}
A difference between the shift of the mean in the solution of the KPZ equation and in the liquid crystal experiment is that they have opposite signs. The latter has the same sign as for the TASEP (totally asymmetric) as we will show. Since the solution of the KPZ equation can be obtained starting from the WASEP (weakly asymmetric), there will be a value of the asymmetry for which the mean has no shift (up to $\Or(t^{-2/3})$). A preliminary Monte-Carlo simulation indicates that this happens for the PASEP height function at the origin for $p=p_c\simeq 0.78$~\cite{SS10b}.
\begin{result}\label{claim0}
In Corollary~\ref{CorShiftPASEP} we show that $p_c$ is the solution of
\begin{equation}
\sum_{\ell=1}^\infty \frac{(1-p_c)^\ell}{p_c^\ell-(1-p_c)^\ell}=\frac12\quad\iff \quad p_c=0.78227\,87862{\ldots}
\end{equation}
\end{result}
We first determine an analytic formula for the shift of the distribution of a \emph{tagged particle} (see Proposition~\ref{PropPASEP}). The shift turns out to be a function of the macroscopic particle number. However, when we switch back to the height function representation, the shift becomes again independent of the macroscopic position (the $\xi$ in \eqref{eqIntro1}).

The other results concern the first order correction to the limiting distribution function and density. Let us illustrate it for the PNG droplet. The other cases are analogue, but instead of the GUE one has for example the GOE Tracy-Widom distribution. The height function is an integer-valued function, i.e.\ $h_t(0)\in\Z$, and since we always look at it it at a single point, $x=0$, we drop the space dependence in our notation and write $h_t \equiv h_t(0)$. It is known that for some explicit constants $c_1,c_2$, the rescaled variable
\begin{equation}\label{eq2b}
\widetilde h_{t,\rm resc} := \frac{h_t-c_1 t}{c_2 t^{1/3}} \overset{{\cal D}}{\longrightarrow} \zeta\quad \textrm{as }t\to\infty,
\end{equation}
where $\zeta$ is a random variable with (GUE) Tracy-Widom distribution $F_{\rm GUE}$, i.e.,
\begin{equation}
\widetilde F_t(s):=\Pb(\widetilde h_{t,\rm resc}\leq s) \to F_{\rm GUE}(s)\quad \textrm{as }t\to\infty.
\end{equation}
Remark that $\widetilde F_t$ is piecewise constant over intervals of length $\delta_t:=1/(c_2 t^{1/3})$.
One expects that
\begin{equation}\label{eq7b}
\widetilde h_{t,\rm resc} = \zeta + \eta \, \delta_t + \Or(\delta_t^{2})\quad\textrm{on }\widetilde I_t:=(\Z-c_1 t) \delta_t.
\end{equation}
where $\eta$ is another random variable a priori not independent from $\zeta$. With (\ref{eq7b}) we mean that
\begin{equation}
\Pb(\widetilde h_{t,\rm resc}\leq s)=\Pb(\zeta + \eta \, \delta_t + \Or(\delta_t^{2})\leq s)
\end{equation}
for $s\in \widetilde I_t$.

But what is the nature of $\eta$? The surprising result is that for the models we consider, $\eta$ is a \emph{deterministic constant} and therefore \emph{independent} from $\zeta$ (see Section~\ref{SectPNGTASEP} for PNG and TASEP, Section~\ref{SectPASEP} for PASEP). This implies the following.
\begin{result}\label{claim1}
Let us denote by $\delta_t:=c_2^{-1} t^{-1/3}$ the discrete lattice width where $\widetilde h_{t,\rm resc}$ lives.
There exists a constant $\eta$ such that
\begin{equation}\label{eq8b}
\widetilde F_t(s)=\Pb\bigl(\widetilde h_{t,\rm resc}\leq s\bigr)= F_{\rm GUE}(s-\eta \, \delta_t) + \Or(\delta_t^2)
\end{equation}
for all $s\in \widetilde I_t=(\Z-c_1 t)\delta_t$.
\end{result}
For PNG and TASEP the shift does not depend on the chosen macroscopic position, but this property is not generic and might depend on the chosen observable too, as shown by the result on PASEP. (This non-universality of $\eta$ is quite intuitive, since $\eta$ is a correction term on the microscopic scale, thus model-dependent.) Consequently, by shifting the height function $h_t$ by the constant $\eta$ as in (\ref{eq8b}), the convergence of the distribution function to $F_{\rm GUE}$ is of order $\Or(t^{-2/3})$. If $\eta$ was not independent from $\zeta$, then one would have a convergence only of order $\Or(t^{-1/3})$ instead.

In the domain of random matrices, similar results have been obtained for the Gaussian and Laguerre Unitary Ensembles~\cite{Cho06,Kar06}. However, in those cases the analyzed random variable were continuous. This differs from the random variables of the models considered here, since before rescaling they live on $\Z$, while after the rescaling (\ref{eqIntro1}) they still live on a discrete lattice of width $\delta_t$. The discreteness becomes irrelevant for the universal statements, but at first order it can not be neglected for the fit with the limiting distribution function and density. Indeed, the shift needed to have a fit with accuracy of order $\Or(t^{-2/3})$ is not the same for the density as for distribution function. In order to see this feature, consider the slightly modified scaling of the height function
\begin{equation}\label{eq11}
h_{t,\rm resc}:=(h_t-c_1 t-a)\delta_t
\end{equation}
where $a\in\R$ is a given constant. Further, set
\begin{equation}
F_t(s):=\Pb(h_{t,\rm resc}\leq s),
\end{equation}
and
\begin{equation}
p_t(s):=\delta_t^{-1}(F_t(s)-F_t(s-\delta_t)).
\end{equation}
\begin{result}\label{claim2}
With the choice $a:=\eta+\tfrac12$ we have
\begin{equation}
p_t(s)=F'_{\rm GUE}(s)+\Or(\delta_t^2).
\end{equation}
for all $s\in I_t:=(\Z-c_1 t-a)\delta_t$.
\end{result}
These results are discussed in Section~\ref{SectDiscreteness} and used for the fits of the simulations of TASEP in Section~\ref{SectPNGTASEP}.

\begin{remark}
Result~\ref{claim2} is generic and does not depend on the fact that our distributions is expressed by a Fredholm determinant, but is a consequence of the $\Or(\delta_t^2)$ error for the centered discrete derivative.
\end{remark}

\begin{remark}
With the scaling (\ref{eq11}), Result~\ref{claim1} writes
\begin{equation} \label{eq9}
F_t(s)=F_{\rm GUE}(s+\tfrac12 \delta_t)+\Or(\delta_t^2), \quad s \in I_t,
\end{equation}
while the scaling (\ref{eq2b}) yields
\begin{equation}
\widetilde F_t(s)=F_{\rm GUE}(s+\tfrac12 \delta_t-a\delta_t)+\Or(\delta_t^2), \quad s \in \widetilde I_t, \label{eq12}
\end{equation}
and
\begin{equation}
\widetilde p_t(s)=F'_{\rm GUE}(s-a\delta_t)+\Or(\delta_t^2), \quad s \in \widetilde I_t, \label{eq13}
\end{equation}
where $\widetilde p_t(s):=\delta_t^{-1}(\widetilde F_t(s)-\widetilde F_t(s-\delta_t))$.
\end{remark}

In view of these results, we carried out a simulation for TASEP with time $t=1000$. As observable we used a tagged particle. The dots in Figure~\ref{FigStepDensity} represent $s\mapsto p_t(s)$ for $s \in I_t$ and $a=\eta+\tfrac12$, which is well approximated by the solid line $s \mapsto F_{\rm GUE}'(s)$ as predicted by Result~\ref{claim2}. As comparison, the dashed line is the unshifted density $s \mapsto F_{\rm GUE}'(s-a\delta_t)$ (see (\ref{eq13})), i.e., the fit obtained with $a=0$.

The same applies to the distribution function. The dots in Figure~\ref{FigStepCDF} are the plot of $s \mapsto F_t(s)$ for $s \in I_t$ and $a=\eta+\tfrac12$. The dashed line is the predicted limiting distribution function with scaling (\ref{eq2b}), be $s \mapsto F_{\rm GUE}(s) = F_{\rm GUE}(s+\tfrac12\delta_t - a\delta_t)$ (see (\ref{eq12})). The fit suggested by Result~\ref{claim1} is the solid line, $s \mapsto F_{\rm GUE}(s+\tfrac12\delta_t)$, which indeed is a better fit.

In the same way we fit Figures~\ref{FigFlatDensity} and~\ref{FigFlatCDF} with the difference that the limiting distribution function is $s\mapsto F_{\rm GOE}(2s)$.

\medskip

Finally, the shift used in $h_{t,\rm resc}$ is the same needed to have a convergence of the moments, and consequently of the variance, skewness, kurtosis of order $\Or(t^{-2/3})$. The following result is discussed in Section~\ref{sectmoments}.
\begin{result}\label{claim3}
We have
\begin{equation}
\E(h_{t,\rm resc}^m)=\E(\zeta^m)+\Or(\delta_t^2)
\end{equation}
for all $m \in \N$.
\end{result}
Remark that if $\eta$ was not independent from $\zeta$, the convergence of the variance, skewness, and kurtosis would still be of order $\Or(t^{-1/3})$.

\subsubsection*{Acknowledgments}
The authors wish to thank Herbert Spohn and Tomohiro Sasamoto for discussions about this problem and Jinho Baik concerning random matrix results. Further, this work is supported by the German Research Foundation via the SFB611-A12 project.

\section{Strategy and effects of the discreteness}\label{SectDiscreteness}
In this section we present the strategy used to get the results. We discuss the effects of the intrinsic discreteness of the models on the fitting functions and on the moments, since it is relevant at first order. Finally, we explain how to fit data coming from an experiment.

\subsection{On the fitting functions (Results~\ref{claim1} and~\ref{claim2})}
For the PNG model and the TASEP, the strategy of getting (\ref{eq2}) is the following\footnote{Mathematically, we get a weaker result, but to illustrate what really happens let us assume that one has (\ref{eq2}). What is missing are explicit bounds on the decay of the kernels, which can be obtained by standard asymptotic analysis; the ingredients like the steep descent paths are all already contained in previous papers. For TASEP we illustrate the results with a simulation for time $t=1000$.}. In these cases, the distribution function of $h_t$ can be expressed as a (discrete) Fredholm determinant with kernel $K_t$,
\begin{equation}
\Pb(h_t\leq x)=\det(\Id-K_t)_{\ell^2(\{x+1,x+2,\ldots\})}, \quad x\in \Z.
\end{equation}
For some constant $a \in \R$, the rescaled random variable
\begin{equation}\label{eq:1}
h_{t,\rm resc} := (h_t-c_1 t -a) \delta_t\quad\textrm{with }\delta_t=c_2^{-1}t^{-1/3}
\end{equation}
lives on $I_t:=(\Z-c_1t-a)\delta_t$. According to the scaling in (\ref{eq:1}), we define the rescaled kernel $K_{t,\rm resc}$ as
\begin{equation}\label{eq1}
K_{t,\rm resc}(s_1,s_2):= \delta_t^{-1} K_t(c_1 t+a+s_1 \delta_t^{-1},c_1 t+a+s_2 \delta_t^{-1})
\end{equation}
so that the distribution function $F_t$ defined by
\begin{equation}
F_t(s):=\Pb(h_t\leq c_1 t+s \delta_t^{-1} +a), \quad s \in \R,
\end{equation}
can be written as a Fredholm determinant on $\ell^2([s+\delta_t,\infty)\cap I_t,\delta_t\nu)$ with $\nu$ the point measure on $I_t$,
\begin{equation}\label{eq7}
\begin{aligned}
F_t(s)&=\det(\Id-K_{t,\rm resc})_{\ell^2([s+\delta_t,\infty)\cap I_t,\delta_t\nu)} \\
& = \sum_{n=0}^\infty \frac{(-1)^n}{n!} \sum_{x_1,\ldots,x_n\in J_t} \delta_t^n \det(K_{t, \rm resc}(x_i,x_j))_{1\leq i,j\leq n}.
\end{aligned}
\end{equation}
Note that $F_t$ and the Fredholm determinant in (\ref{eq7}) are piecewise constant functions, with jumps for values of $s$ in the lattice $I_t$. The next step is to show that for $s_1,s_2\in I_t$ and a well-chosen $a\in\R$,
\begin{equation}\label{eq3.5}
K_{t,\rm resc}(s_1,s_2)=K(s_1,s_2)+\delta_t K_{\rm asym}(s_1,s_2)+ \Or(\delta_t^2)
\end{equation}
where $K$ is a symmetric and $K_{\rm asym}$ an antisymmetric kernel. Then, it follows that\footnote{On a rigorous level, one needs to verify that (a) the $\Or(\delta_t^2)$ in (\ref{eq3.5}) is an operator with with $1$-norm of order $\Or(\delta_t^2)$ and (b) $(\Id-\chi_s K \chi_s)^{-1}\chi_s K_{\rm asym}\chi_s$ is trace-class.}
\begin{equation}
\begin{aligned}
F_t(s)&=\det(\Id-K_{t,\rm resc})_{\ell^2([s+\delta_t,\infty)\cap I_t,\delta_t\nu)}\\
&=\det(\Id-K)_{\ell^2([s+\delta_t,\infty)\cap I_t,\delta_t\nu)}\\
&\quad\times \left(1+\delta_t\Tr((\Id-\chi_sK\chi_s)^{-1})\chi_sK_{\rm asym}\chi_s)+ \Or(\delta_t^2) \right)
\end{aligned}
\end{equation}
where $\chi_s$ is the projection onto $[s+\delta_t,\infty)\cap I_t$. The operator under the trace is antisymmetric, therefore its trace is zero and
\begin{equation}\label{eq2}
F_t(s)=\det(\Id-K)_{\ell^2([s+\delta_t,\infty)\cap I_t,\delta_t\nu)}\left(1+\Or(\delta_t^{2}) \right), \quad s \in \R.
\end{equation}
If we denote by
\begin{equation}\label{eq4}
F(s):=\det(\Id-K)_{L^2((s,\infty))}
\end{equation}
the limiting distribution of $F_t(s)$ taken as a Fredholm on $L^2((s,\infty))$, then by Lemma~\ref{LemDiscrVsCont} below, we get
\begin{equation}
F_t(s)=F(s+\tfrac{1}{2}\delta_t)+\Or(\delta_t^2),
\end{equation}
for $s\in I_t$, and by the argument below (that gives Result~\ref{claim2}), one finally obtains
\begin{equation}
\quad p_t(s)=F'(s)+\Or(\delta_t^2).
\end{equation}
In Section~\ref{SectPNGTASEP} we derive (\ref{eq3.5}) for the PNG model and the TASEP.

Let us explain how to get Result~\ref{claim2} without the need of Fredholm determinant representations.
Assume that there exists a constant $\gamma$ such that
\begin{equation}
F_t(s)=F(s+\gamma \delta_t) + \delta_t^2 Q(s)+\Or(\delta_t^3),\quad s\in I_t,
\end{equation}
with $F\in C^2$ and $Q\in C^1$. Then using Taylor expansion we readily obtain
\begin{equation}
p_t(s)=F'(s) + \frac{\delta_t}{2}(\gamma^2-(1-\gamma)^2) F''(s)+\Or(\delta_t^2),\quad s\in I_t.
\end{equation}
Therefore, if $\gamma=1/2$, then $p_t(s)-F'(s)=\Or(\delta_t^2)$ for $s\in I_t$, while the approximation would be only of order $\delta_t$ if $\gamma \neq 1/2$.

In our case, see Corollary~\ref{CorDistrVSdensity}, we have $\gamma=1/2$ which is a consequence of the following lemma.
\begin{lem}\label{LemDiscrVsCont}
Assume that the kernel $K$ satisfies\footnote{With $\partial_i$ we mean the derivative with respect to the $i$th entry of the function. The assumption (\ref{eq6}) holds for the Airy kernels, see Lemma~\ref{LemAiryKernels}.}
\begin{equation} \label{eq6}
\max\{\lvert K(x_1,x_2)\rvert, \lvert \partial_i K(x_1,x_2) \rvert, \lvert \partial_i\partial_j K(x_1,x_2) \rvert\} \leq C e^{-c(x_1+x_2)}
\end{equation}
for some constants $C,c>0$, for all $x_1,x_2\in (s,\infty)$ and $i,j \in \{1,2\}$. Let $\delta_t$ be as above the lattice width. Then\footnote{For the Airy kernels it is easy to improve $\Or(t^{-2/3} e^{-s})$ to $\Or(t^{-2/3} e^{-\max\{s,0\}})$. However, getting a rigorous \emph{good} bound for the error as $s\to-\infty$ is a much more difficult task (this would be needed for a rigorous proof of the convergence of the moments).}
\begin{equation}\label{eq3}
\left|\det(\Id-K)_{L^2((s+\delta_t/2,\infty))}-\det(\Id-K)_{\ell^2([s+\delta_t,\infty)\cap I_t,\delta_t\nu)}\right| =  \Or(\delta_t^2 e^{-cs}).
\end{equation}
\end{lem}
\begin{proof}
Let us set $J_t:=I_t\cap[s+\delta_t,\infty)$. Then, we have
\begin{equation}
\det(\Id-K)_{\ell^2(J_t,\delta_t\nu)} = \sum_{n=0}^\infty \frac{(-1)^n}{n!}  \sum_{x_1,\ldots,x_n\in J_t} \delta_t^n \det(K(x_i,x_j))_{1\leq i,j\leq n}
\end{equation}
and
\begin{equation}
\det(\Id-K)_{L^2((s+\delta_t/2,\infty))} = \sum_{n=0}^\infty \frac{(-1)^n}{n!} \int\limits_{(s+\delta_t/2,\infty)^n}\dx^n x\, \det(K(x_i,x_j))_{1\leq i,j\leq n}.
\end{equation}
Equation \eqref{eq3} then follows from Lemma~\ref{LemmaSumIntNdim} with
\begin{equation}
f(x_1,\ldots,x_n):=\det(K(s+\delta_t+x_i,s+\delta_t+x_j))_{1\leq i,j\leq n}
\end{equation}
together with Lemma~\ref{lemBoundKgeneral}.
\end{proof}
A straightforward corollary is the following.
\begin{cor}\label{CorDistrVSdensity}
Assume \eqref{eq2} and \eqref{eq6} to hold. Then, for large $t$, we have (remember that $\delta_t=c_2^{-1}t^{-1/3}$)
\begin{equation}\label{eq10}
F_t(s)=F(s+\tfrac12 \delta_t)+\Or(\delta_t^2)
\end{equation}
for $s\in I_t$.
\end{cor}

\begin{remark}\label{RemarkNoShift}
An equivalent way would be to consider the scaling (\ref{eq1}) without the shift by $a$, i.e.,
\begin{equation}\label{eq1b}
\widetilde K_{t,\rm resc}(s_1,s_2):= \delta_t^{-1} K_t(c_1 t+s_1 \delta_t^{-1},c_1 t+s_2 \delta_t^{-1})=K_{t,\rm resc}(s_1-a\delta_t,s_2-a\delta_t)
\end{equation}
Then, instead of (\ref{eq3.5}) we would have obtained
\begin{equation}
\begin{aligned}
\widetilde K_{t,\rm resc}(s_1,s_2) =&K(s_1,s_2)-a\delta_t(\partial_1 K(s_1,s_2)+\partial_2 K(s_1,s_2))\\
&+\delta_t \widetilde K_{\rm asym}(s_1,s_2)+ \Or(\delta_t^2).
\end{aligned}
\end{equation}
In the specific case of the Airy kernels  $K_{{\cal A}_2}(x,y):=\int_{\R_+}\dx \lambda \, \Ai(x+\lambda)\Ai(y+\lambda)$ and $K_{{\cal A}_1}(x,y):=\Ai(x+y)$,
\begin{equation}
\partial_1 K_{{\cal A}_2}(s_1,s_2)+\partial_2 K_{{\cal A}_2}(s_1,s_2) =-\Ai(s_1)\Ai(s_2),
\end{equation}
and
\begin{equation}
\partial_1 K_{{\cal A}_1}(s_1,s_2)+\partial_2 K_{{\cal A}_1}(s_1,s_2) = 2{\Ai}'(s_1+s_2).
\end{equation}
\end{remark}

\subsection{On the moments (Result~\ref{claim3})}\label{sectmoments}
Another consequence of the constant shift by $a$ is that all finite moments of $F_t$ converge as fast as $t^{-2/3}$. Without the shift, the first moment would converge only as fast as $t^{-1/3}$, while the variance, skewness, kurtosis would of course not be affected by the shift.
\begin{lem}\label{LemMoments}
Assume that\footnote{Note that this condition is stronger than (\ref{eq10}) and in general not so easy to obtain rigorously.} $F_t(s)=F(s+\tfrac{\delta_t}{2})+\Or(\delta_t^2)G_t(s)$ for $s\in I_t$ such that $F$ has finite $m$th moment (with $F''\in L^1\cap C^0$) and $G_t$ satisfying $\int_{\R}\dx s\, |s|^m |G_t(s)|<\infty$ uniformly in $t$. Then,
\begin{equation}
\int_{\R} s^m \, \dx F_t(s)=\int_{\R} s^m \, \dx F(s) + \Or(\delta_t^2)
\end{equation}
for all $m\in\N$.
\end{lem}
\begin{proof}
Let us set $I_t=I_t^+ \cup I_t^-$ where $I_t^\pm=I_t\cap \R_{\pm}$, $w=\sup I_t^-+\delta_t=\inf I_t^+$, and $\tilde I_t^\pm:=I_t^\pm+\delta_t/2$. Then, for any $m\geq 1$,
\begin{equation}\label{eq8}
\begin{aligned}
\int_{\R} & s^m \, \dx F_t(s)\equiv \delta_t \sum_{s\in I_t} s^m \, p_t(s)= \sum_{s\in I_t} s^m (F_t(s)-F_t(s-\delta_t)) \\
&= \sum_{s\in I_t^+} s^m (F_t(s)-1) - \sum_{s\in I_t^+} (s+\delta_t)^m (F_t(s)-1)+ \sum_{s\in I_t^-} s^m F_t(s)\\
&- \sum_{s\in I_t^-} (s+\delta_t)^m F_t(s)+ w^m F_t(w-\delta_t)-w^m (F_t(w-\delta_t)-1)\\
&=w^m+\sum_{s\in I_t^+} (s^m-(s+\delta_t)^m) (F_t(s)-1)+ \sum_{s\in I_t^-} (s^m-(s+\delta_t)^m) F_t(s).
\end{aligned}
\end{equation}
Now we set $\tilde s=s+\tfrac{\delta_t}{2}$ so that
\begin{equation}
(s^m-(s+\delta_t)^m)=(\tilde s-\tfrac{\delta_t}{2})^m-(\tilde s+\tfrac{\delta_t}{2})^m=-m \delta_t \tilde s^{m-1} + \Or(\delta_t^3).
\end{equation}
Then using $F_t(\tilde s-\tfrac{\delta_t}{2})=F(\tilde s)+\Or(\delta_t^2)G_t(\tilde s)$ we obtain
\begin{equation}
\begin{aligned}
\eqref{eq8}&=w^m -m \delta_t \sum_{\tilde s\in \tilde I_t^+} \tilde s^{m-1} (F(\tilde s)-1)
-m\delta_t \sum_{\tilde s\in \tilde I_t^-} \tilde s^{m-1} F(\tilde s)+ \Or(\delta_t^2)\\
&=w^m- m \delta_t \int_{w}^\infty \dx s\, s^{m-1} (F(s)-1)
-m\delta_t \int_{-\infty}^{w} \dx s\,  s^{m-1} F(s)+ \Or(\delta_t^2)\\
 &= \int_{\R} s^m \, \dx F(s)+\Or(\delta_t^2).
\end{aligned}
\end{equation}
where we used Lemma~\ref{LemmaSumIntNdim} to approximate the sums by the integrals.
\end{proof}

\subsection{How to fit the experimental data}
For completeness, we explain shortly how to fit the experimental data. We partially follow the description of the Supplementary Notes of \cite{TSSS11} and use their notations. Let us assume that we observed a growth process which is thought to belong to the KPZ class. Let $S=\e\, \Z$, $\e>0$, be a discrete subset of $\R$, where the values of the height function at time $t$, denoted by $h_t$, lives. Let $N\gg 1$ the number of experimental measurements and denote by $\langle \, \cdot \, \rangle$ the empirical average over the $N$ experiments.
Having \eqref{eq7b} in mind, we expect to have
\begin{equation}
h_t\simeq v_\infty \, t +(\Gamma t)^{1/3} \zeta + a
\end{equation}
where $\zeta$ is a GUE (resp.\ GOE) Tracy-Widom distributed random variable for curved (resp.\ flat) limit shape, $v_\infty$ the asymptotic growth velocity and $a$ a constant.
\begin{itemize}
\item[(1)] \emph{Determine the asymptotic growth velocity $v_\infty$}. Using
\begin{equation}
\frac{\dx \langle h_t\rangle}{\dx t} \simeq v_\infty+b \, t^{-2/3},\quad b=\Gamma^{1/3}\E(\zeta)/3.
\end{equation}
one obtains $v_\infty$ from the plot $(t^{-2/3},\frac{\dx \langle h_t\rangle}{\dx t})$.
\item[(2)] \emph{Verify the fluctuation scaling exponent and the fluctuation amplitude $\Gamma$}. With a log-log plot we can verify if the power $2/3$ in
\begin{equation}
\langle(h_t - \langle h_t\rangle)^2\rangle \simeq (\Gamma t)^{2/3} \Var(\zeta)
\end{equation}
holds and at the same time measure the constant $\Gamma\neq 0$.
\item[(3)] \emph{Determine the shift parameter $a$}. Consider the standard KPZ scaling
$\widetilde h_{t,\rm resc} := (h_t - v_\infty\, t)/(\Gamma \,t)^{1/3}$.
Then, $a$ is measured according to the relation
\begin{equation}
\langle \widetilde h_{t,\rm resc} \rangle -\E(\zeta) \simeq a \, (\Gamma t)^{-1/3}.
\end{equation}
\end{itemize}
Now that we have determined $v_\infty$, $\Gamma$ and $a$, we fit the data vs.\ the theoretical predictions. We set $h_{t,\rm resc}:=\widetilde h_{t,\rm resc}-a \,(\Gamma t)^{-1/3}$.
\begin{itemize}
\item[(4.1)] \emph{Density}: We do a plot of the frequencies of $h_{t,\rm resc}$ with the set \mbox{$(S-v_\infty t-a)/(\Gamma t)^{1/3}$} in the abscissa axis. Then we compare this with the graph of the Tracy-Widom densites $s\mapsto F'(s)$.
\item[(4.2)] \emph{Distribution function}: We do a plot of the cumulated frequencies of $h_{t,\rm resc}$ with the set $(S-v_\infty t-a)/(\Gamma t)^{1/3}$ in the abscissa axis. Then we compare this with the graph of the (shifted) Tracy-Widom distribution function $s\mapsto F(s+\tfrac12\e/ (\Gamma t)^{1/3}))$.
\end{itemize}

\section{PNG and TASEP}\label{SectPNGTASEP}
In this section we determine the value of the order $1$ shifts for the PNG and TASEP models, both with flat and curved geometry.

\subsection{Flat PNG}
In~\cite{BFS07b} the formula for the height function $h_t$ at time $t$ for the flat PNG was obtained\footnote{For the one-point distribution there exists also a formulation in terms of Fredholm Pfaffian.}. It is shown that
\begin{equation}
\Pb(h_t(0)\leq H)=\det(\Id-K_t^{\rm flatPNG})_{\ell^2(\{H+1,H+2,\ldots\})}
\end{equation}
with
\begin{equation}
K_t^{\rm flat PNG}(x_1,x_2)=J_{x_1+x_2}(4t)=\frac{1}{2\pi\I}\oint_{\Gamma_0}\dx z\, \frac{e^{2t(z-z^{-1})}}{z^{x_1+x_2+1}},
\end{equation}
where $J_n$ is the standard Bessel function (we use the conventions of~\cite{AS84})\footnote{With the notation $\Gamma_S$, with $S$ a set, we mean any simple counterclockwise oriented path encircling the set $S$.}.
As $t\to\infty$, we consider the scaling
\begin{equation}
H(s)=2t+s(2t)^{1/3}\in \N \quad\Rightarrow \quad s\in I_t=(\N-2t) (2t)^{-1/3}
\end{equation}
Under this scaling it is known that~\cite{BFS07b}
\begin{multline}
K_{t,\rm resc}^{\rm flatPNG}(s_1,s_2):=(2t)^{1/3} K_t^{\rm flatPNG}(H(s_1),H(s_2)) \\
\to \Ai(s_1+s_2) = K_{{\cal A}_1}(s_1,s_2)
\end{multline}
as $t\to\infty$ and uniformly for $s_1,s_2$ in bounded sets. Moreover, there are exponential bounds for the decay of $K_{t,\rm resc}^{\rm flatPNG}$ (see, e.g., Appendix A.2 of~\cite{Fer04}) which ensures that we can take the limit $t\to\infty$ inside the Fredholm determinant, leading to\footnote{In~\cite{BFS07b} the result is for joint distributions of the height function at different positions. The one-point distribution was also obtained through its relation with symmetrized permutations~\cite{BR01b}.}
\begin{equation}
\lim_{t\to\infty}\Pb(h_t(0)\leq 2t+s (2t)^{1/3})=\det(\Id-K_{{\cal A}_1})_{L^2((s,\infty))} =F_{\rm GOE}(2s),
\end{equation}
where $F_{\rm GOE}$ is the GOE Tracy-Widom distribution function~\cite{TW96}.

Here we focus on the first order correction of $K_{t,\rm resc}^{\rm flatPNG}$ with respect to $K_{{\cal A}_1}$ and show that it is zero. Since the asymptotic analysis is quite standard (see e.g.\ Lemma 6.1 of~\cite{BF08} for the explanation of general strategy), here and in the next sections we indicate only the important steps.
\begin{prop}
Uniformly for $s_1,s_2$ in a bounded subset of $I_t$,
\begin{equation}
K_{t,\rm resc}^{\rm flatPNG}(s_1,s_2)= K_{{\cal A}_1}(s_1,s_2)+\Or(t^{-2/3}).
\end{equation}
\end{prop}
\begin{proof}
We have
\begin{equation}\label{eqRescFlatPNG}
K_{t,\rm resc}^{\rm flatPNG}(s_1,s_2) = \frac{(2t)^{1/3}}{2\pi \I} \oint_{\Gamma_0} \dx z \, \frac{e^{2t(z-z^{-1})}}{z^{4t+(s_1+s_2)(2t)^{1/3}+1}}.
\end{equation}
The function $z \mapsto z-z^{-1}-2\ln z$ has a double critical point at $z_c=1$. The steepest descent path can be taken to be coming into $z_c$ with an angle $e^{-\pi\I/3}$, leaving with an angle $e^{\pi\I/3}$, and completed by a piece of a circle around zero with a radius strictly larger than $1$. Then, the leading term in the asymptotic of $K_{t,\rm resc}^{\rm flatPNG}$ comes from a $t^{-1/3}$-neighborhood of $z_c$. Setting $z=1+Z (2t)^{-1/3}$ and doing the large $t$ expansion of the integrand in (\ref{eqRescFlatPNG}), one obtains
\begin{multline}
K_{t,\rm resc}^{\rm flatPNG}(s_1,s_2) = \frac{1}{2\pi\I} \int_{\infty e^{-\pi \I/3}}^{\infty e^{\pi \I/3}} \dx Z \, \exp\left(\frac{Z^3}{3} - \zeta Z\right) \\
 \qquad \qquad \times  \left(1-t^{-1/3}\left( \frac{Z}{2^{1/3}} - \frac{\zeta Z^2}{2^{4/3}} + \frac{Z^4}{2^{4/3}}\right)+\Or(t^{-2/3})\right),
\end{multline}
where we have set for simplicity $\zeta:=s_1+s_2$. Using the contour integral representation of the Airy function (\ref{eqAiry}) we have
\begin{multline}
K_{t,\rm resc}^{\rm flatPNG}(s_1,s_2) = \Ai(\zeta) + t^{-1/3} \bigg[\frac{\Ai'(\zeta)}{2^{1/3}} + \frac{\zeta \Ai''(\zeta)}{2^{4/3}}-\frac{\Ai^{(4)}(\zeta)}{2^{4/3}}\bigg]+\Or(t^{-2/3}).
\end{multline}
Finally, using the identity $\Ai''(\zeta)=\zeta \Ai(\zeta)$ of the Airy function one readily gets that the square bracket is equal to zero.
\end{proof}

\subsection{PNG droplet}
The formula for the height function $h_t$ at time $t$ for the PNG droplet was determined in~\cite{PS02}. Let us fix $c\in (-1,1)$. Then,
\begin{equation}
\Pb(h_t(ct)\leq H)=\det(\Id-K_{t,c}^{\rm curvPNG})_{\ell^2(\{H+1,H+2,\ldots\})}
\end{equation}
with
\begin{equation}
K_{t,c}^{\rm curvPNG}(x_1,x_2)=\sum_{\ell\geq 0} J_{x_1+\ell}(2t\sqrt{1-c^2})J_{x_2+\ell}(2t\sqrt{1-c^2}).
\end{equation}
We consider the case $c=0$ (and drop the index $c$) since the general case is simply obtained by replacing $t$ by $t\sqrt{1-c^2}$. An integral representation of the kernel is given by
\begin{equation}
K_t^{\rm curvPNG}(x_1,x_2)=\frac{1}{(2\pi\I)^2}\oint_{\Gamma_0}\dx w\oint_{\Gamma_{0,w}}\dx z \, \frac{e^{2t(z-z^{-1})}}{e^{2t(w-w^{-1})}}\frac{w^{x_2-1}}{z^{x_1}}\frac{1}{z-w}.
\end{equation}
As $t\to\infty$, we consider the scaling
\begin{equation}
H(s)=2t+s t^{1/3}+a\in \N \quad\Rightarrow \quad s\in I_t=(\N-2t-a) t^{-1/3}
\end{equation}
for a $t$-independent constant $a$ to be specified later. In~\cite{PS02} it is proven that the rescaled kernel converges to the Airy kernel, namely
\begin{multline}
K_{t,\rm resc}^{\rm curvPNG}(s_1,s_2):=t^{1/3} K_t^{\rm curvPNG}(H(s_1),H(s_2))\\
\to \int_{\R_+}\dx \lambda \Ai(s_1+\lambda)\Ai(s_2+\lambda)=K_{{\cal A}_2}(s_1,s_2),
\end{multline}
as $t\to\infty$ and uniformly for $s_1,s_2$ in bounded sets. Moreover, exponential bounds for the decay of $K_{t,\rm resc}^{\rm curvPNG}$ ensure that\footnote{The extension to joint distributions was obtained in~\cite{PS02}, while the one-point result is reported in~\cite{PS00} using a mapping to the Poissonized longest increasing subsequence problem, which was already solved in~\cite{BDJ99}.}
\begin{equation}
\lim_{t\to\infty}\Pb(h_t(0)\leq 2t+s t^{1/3}+a)=\det(\Id-K_{{\cal A}_2})_{L^2((s,\infty))} = F_{\rm GUE}(s),
\end{equation}
where $F_{\rm GUE}$ is the GUE Tracy-Widom distribution function~\cite{TW94}.

The first order correction of $K_{t,\rm resc}^{\rm curvPNG}$ with respect to $K_{{\cal A}_2}$ is the following.
\begin{prop}
Uniformly for $s_1,s_2$ in a bounded subset of $I_t$, with the choice $a=1/2$,
\begin{equation}
K_{t,\rm resc}^{\rm curvPNG}(s_1,s_2)=K_{{\cal A}_2}(s_1,s_2)+\Or(t^{-2/3}).
\end{equation}
\end{prop}
\begin{proof}
The rescaled kernel is
\begin{equation}\label{eqRescPNGdroplet}
K_{t,\rm resc}^{\rm curvPNG}(s_1,s_2)=\frac{t^{1/3}}{(2\pi\I)^2}\oint_{\Gamma_0}\dx w\oint_{\Gamma_{0,w}}\dx z \, \frac{e^{t(z-z^{-1})}}{e^{t(w-w^{-1})}}\frac{w^{2t+s_2t^{1/3}+a-1}}{z^{2t+s_1t^{1/3}+a}}\frac{1}{z-w}.
\end{equation}
Here, we have to integrate over two contours, the one in the $z$-variable enclosing the contour in the $w$-variable. The steepest descent path for $z$ can be taken as in the flat case, the one for $w$ leaves the critical points with an angle $e^{2\pi\I/3}$ and then is completed by a piece of a circle of radius strictly smaller than $1$. Doing the change of variables
\begin{equation}
z=1+t^{-1/3}Z,\quad w=1+t^{-1/3}W,
\end{equation}
we eventually get
\begin{multline}
K_{t,\rm resc}^{\rm curvPNG}(s_1,s_2)=\frac{1}{(2\pi\I)^2}\int_{\infty e^{-2\pi\I/3}}^{\infty e^{2\pi\I/3}}\dx W
\int_{\infty e^{-\pi\I/3}}^{\infty e^{\pi\I/3}}\dx Z \frac{e^{Z^3/3-s_1 Z}}{e^{W^3/3-s_2 W}}\frac{1}{Z-W}\\
\times \bigg(1 +t^{-1/3}\biggl[\frac{s_1 Z^2-s_2 W^2}{2} - \frac{Z^4-W^4}{4}
-a Z + (a-1)W \biggr]+\Or(t^{-2/3})\bigg),
\end{multline}
where the integration paths do not intersect. At this point we see that setting $a =1/2$ the first order term is antisymmetric. In particular we can choose the paths to satisfy $\Re Z >\Re W$ and use $\frac{1}{Z-W}=\int_0^\infty\dx\lambda e^{-\lambda(Z-W)}$ to get
\begin{equation}
K_{t,\rm resc}^{\rm curvPNG}(s_1,s_2) = K_{{\cal A}_2}(s_1,s_2)+t^{-1/3}\bigl(P(s_1,s_2)-P(s_2,s_1)\bigr) +\Or(t^{-2/3})
\end{equation}
with $P$ given by
\begin{equation}
P(s_1,s_2)=\frac{1}{2}\int_0^\infty \dx\lambda \Ai(s_2+\lambda) \left[\frac{\dx}{\dx s_1}+ s_1 \frac{\dx^2}{\dx s_1^2} - \frac{\dx^4}{\dx s_1^4}\right] \Ai(s_1+\lambda).
\end{equation}
Using $\Ai''(x)=x\Ai(x)$ and integration by parts one then shows $P(s_1,s_2)=P(s_2,s_1)$.
\end{proof}

Without the shift by $a$ in the scaling, the result would have been
\begin{equation}
K_{t,\rm resc}^{\rm curvPNG}(s_1,s_2) = K_{{\cal A}_2}(s_1,s_2)-\tfrac12 t^{-1/3}(-\Ai(s_1)\Ai(s_2))) +\Or(t^{-2/3})
\end{equation}
from which we can read off the shift $a=1/2$, compare with Remark~\ref{RemarkNoShift}.

\begin{remark}
The shift by $a=1/2$ is actually independent of $c\in (-1,1)$, which is due to the fact that it is built up during the first stages of the growth process and for large $t$ it converges to $1/2$. Therefore for large $t$, the shift at time $t\sqrt{1-c^2}$ is the same as for the model at time $t$.
\end{remark}

\subsection{TASEP with alternating initial condition}
Now consider TASEP with alternating initial condition, $x_k(t=0)=-2 k$, $k\in \Z$. The joint distribution of particle positions for this initial condition has been determined in~\cite{BFPS06}. For one particle (here $-x_n(t)$ plays the role of $h_t$), we have
\begin{equation}
\Pb(x_{n}(t)\geq X)=\det(\Id-K_{t,n}^{\rm flatTASEP})_{\ell^2(\{\ldots,X-2,X-1\})}
\end{equation}
with
\begin{equation}
K_{t,n}^{\rm flatTASEP}(x_1,x_2)=\frac{-1}{2\pi\I}\oint_{\Gamma_1}\dx z\, e^{t(1-2z)}\frac{z^{2n+x_2}}{(1-z)^{2n+x_1+1}}.
\end{equation}
As $t\to\infty$, we consider the scaling
\begin{equation}\label{eq3.37}
X(s)=-2n+\tfrac{1}{2}t-s t^{1/3}-a\in \Z \quad\Rightarrow \quad s\in I_t=(\Z-t/2+a) t^{-1/3}
\end{equation}
for a $t$-independent constant $a$ to be specified later.
It is known\footnote{The prefactor $2^{X(s_2)-X(s_1)}$ is just a conjugation, which does not change the underlying determinantal point process, but it is needed to have a well-defined limit.} that~\cite{BFPS06}
\begin{multline}
K_{t,\rm resc}^{\rm flatTASEP}(s_1,s_2):=2^{X(s_2)-X(s_1)} t^{1/3} K_t^{\rm flatTASEP}(X(s_1),X(s_2))\\
\to K_{{\cal A}_1}(s_1,s_2)
\end{multline}
as $t\to\infty$ and uniformly for $s_1,s_2$ in bounded sets. Moreover, exponential bounds for the decay of $K_{t,\rm resc}^{\rm flatTASEP}$ ensure that
\begin{equation}
\lim_{t\to\infty}\Pb(x_0(t)\geq t/2-s t^{1/3}-a)=\det(\Id-K_{{\cal A}_1})_{L^2((s,\infty))} = F_{\rm GOE}(2s).
\end{equation}

\begin{figure}
\begin{center}
\psfrag{S}{$s$}
\psfrag{F}{}
\psfrag{1.0}[][][0.8]{$1.0$}
\psfrag{0.8}[][][0.8]{$0.8$}
\psfrag{0.6}[][][0.8]{$0.6$}
\psfrag{0.4}[][][0.8]{$0.4$}
\psfrag{0.2}[][][0.8]{$0.2$}
\psfrag{0.0}[][][0.8]{$0$}
\psfrag{0.5}[][][0.8]{$0.5$}
\psfrag{-0.5}[][][0.8]{$-0.5\hspace{0.8em}$}
\psfrag{-1.0}[][][0.8]{$-1.0\hspace{0.8em}$}
\psfrag{-1.5}[][][0.8]{$-1.5\hspace{0.8em}$}
\psfrag{-2.0}[][][0.8]{$-2.0\hspace{0.8em}$}
\includegraphics[height=10cm,angle=-90]{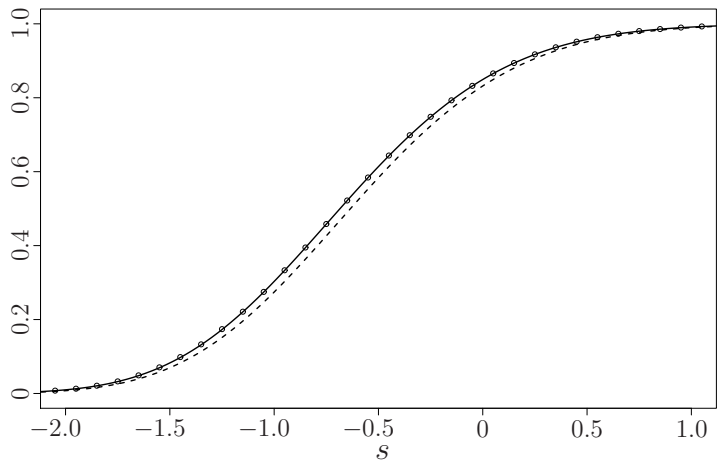}
\caption{Distribution function of $x_{n=[t/4]}(t)$ for TASEP with alternating initial conditions and $t=1000$. The number of runs is $10^6$. The dots are the plot of \mbox{$(s\in I_t,\Pb(x_{[t/4]}(t)\geq X(s)))$} with $\delta_t=t^{-1/3}$ and $a=1/2$. The solid line is \mbox{$(s,F_{\rm GOE}(2s+\tfrac12\delta_t))$} while the dashed line is \mbox{$(s,F_{\rm GOE}(2s+\tfrac12\delta_t-a\delta_t)))$}, where the shift by $a\delta_t$ for dashed line follows from the definition $I_t$, see (\ref{eq3.37}).}
\label{FigFlatCDF}
\end{center}
\end{figure}

\begin{figure}
\begin{center}
\psfrag{S}{$s$}
\psfrag{P}{}
\psfrag{0.6}[][][0.8]{$0.6$}
\psfrag{0.5}[][][0.8]{$0.5$}
\psfrag{0.4}[][][0.8]{$0.4$}
\psfrag{0.3}[][][0.8]{$0.3$}
\psfrag{0.2}[][][0.8]{$0.2$}
\psfrag{0.1}[][][0.8]{$0.1$}
\psfrag{1.0}[][][0.8]{$1.0$}
\psfrag{0.0}[][][0.8]{$0$}
\psfrag{-0.5}[][][0.8]{$-0.5\hspace{0.8em}$}
\psfrag{-1.0}[][][0.8]{$-1.0\hspace{0.8em}$}
\psfrag{-1.5}[][][0.8]{$-1.5\hspace{0.8em}$}
\psfrag{-2.0}[][][0.8]{$-2.0\hspace{0.8em}$}
\includegraphics[height=10cm,angle=-90]{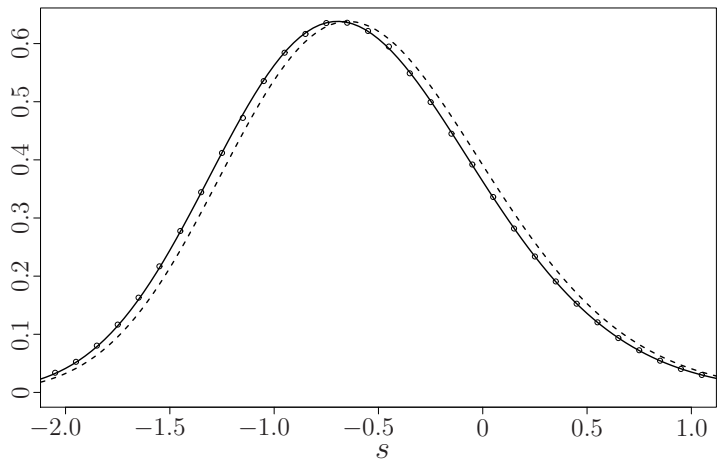}
\caption{Probabilities for $x_{n=[t/4]}(t)$ for TASEP with alternating initial conditions and $t=1000$. The number of runs is $10^6$. The dots are the plot of \mbox{$(s\in I_t,\Pb(x_{[t/4]}(t)= X(s)) \delta_t^{-1})$} with $\delta_t=t^{-1/3}$ and $a=1/2$. The solid line is \mbox{$(s,2F'_1(2s))$} while the dashed line is \mbox{$(s,2F'_1(2s-a\delta_t))$}.}
\label{FigFlatDensity}
\end{center}
\end{figure}

The first order correction of $K_{t,\rm resc}^{\rm flatTASEP}$ with respect to $K_{{\cal A}_1}$ is given as follows.
\begin{prop}
Uniformly for $s_1,s_2$ in a bounded subset of $I_t$, with the choice $a=1/2$, it holds
\begin{equation}
K_{t,\rm resc}^{\rm flatTASEP}(s_1,s_2)= K_{{\cal A}_1}(s_1,s_2)+t^{-1/3}K_{\rm asym}(s_1,s_2)+\Or(t^{-2/3}),
\end{equation}
where $K_{\rm asym}(s_1,s_2)=\frac12(s_2^2-s_1^2)\Ai(s_1+s_2)$.
\end{prop}
\begin{proof}
The rescaled kernel reads
\begin{equation}\label{eqRescFlatTASEP}
K_{t,\rm resc}^{\rm flatTASEP}(s_1,s_2)=-\frac{t^{1/3}}{2 \pi\I}\frac{2^{s_1 t^{1/3}}}{2^{s_2 t^{1/3}}} \oint_{\Gamma_1} \dx z \, e^{t(1-2z)} \frac{z^{t/2-s_2 t^{1/3}-a}}{(1-z)^{t/2-s_1 t^{1/3}-a+1}}
\end{equation}
The function $z \mapsto 1 -2z + \frac{1}{2} \ln \frac{z}{1-z}$ has a double critical point at $z_c=1/2$. We choose as steep descent path the one coming into $z_c$ with angle $e^{\I\pi/3}$, leaving with angle $e^{-\I\pi/3}$, and continued by a piece of a circle around $1$ with radius $1/2$. Setting $Z=t^{1/3}(2z-1)$, we get
\begin{multline}
K_{t,\rm resc}^{\rm flatTASEP}(s_1,s_2) = \frac{1}{2\pi\I} \int_{\infty e^{-\I\pi/3}}^{\infty e^{\I\pi/3}} \dx Z \, e^{Z^3/3 - (s_1+s_2) Z} \\
\times \left(1+ t^{-1/3} \left((1-2a)Z+\frac{s_2-s_1}{2} Z^2\right) + \Or(t^{-2/3})\right).
\end{multline}
Thus we see that in order to make the first order correction antisymmetric we need to choose $a=1/2$. With this choice,
\begin{equation}
K_{t,\rm resc}^{\rm flatTASEP}(s_1,s_2)= \Ai(s_1+s_2)+\tfrac12(s_2-s_1)\Ai''(s_1+s_2)t^{-1/3} + \Or(t^{-2/3})
\end{equation}
and the statement follows using $\Ai''(s_1+s_2)=(s_1+s_2)\Ai(s_1+s_2)$.
\end{proof}

\subsection{TASEP with step initial condition}
Now consider TASEP with step initial condition, $x_k(0)=-k$, $k=1,2,\ldots$ The joint distribution of particle positions for this initial condition can be found for example (as special case) in~\cite{BF07}. For one particle, we have\footnote{The formula for the one-point distribution can be also given by a last passage percolation model, which can be analyzed by determinantal line ensembles leading to the Laguerre kernel~\cite{Jo00b}. Joint distributions for the related last passage model can be determined via Schur process~\cite{Ok01,Jo00b,Jo03b}.}
\begin{equation}
\Pb(x_{n}(t)\geq X)=\det(\Id-K_{t,n}^{\rm stepTASEP})_{\ell^2(\{\ldots,X-2,X-1\})}
\end{equation}
with
\begin{equation}
K_{t,n}^{\rm stepTASEP}(x_1,x_2)=\frac{1}{(2\pi\I)^2} \oint_{\Gamma_0}\dx z \oint_{\Gamma_1} \dx w \, \frac{e^{tz}}{e^{tw}} \left(\frac{1-z}{1-w}\right)^{n} \frac{w^{n+x_2}}{z^{n+x_1+1}}\frac{1}{z-w}
\end{equation}
Now consider a particle that at time $t$ is in the ``rarefaction fan'', i.e., it is in the region with decreasing density strictly between $0$ and $1$. Such particles have a particle number $n=\sigma t$ for some $\sigma \in (0,1)$. Thus, we define for a couple $(n,t)$ the value of $\sigma:=n/t$ and we assume that this value for large $n$ and $t$ is clearly away both from $0$ and $1$. Then, the scaling for $t\to\infty$ is given by
\begin{equation}
n=\sigma t\in \N,\quad X(s)=-n+(1-\sqrt{\sigma })^2t-s c_2 t^{1/3}-a\in \Z,
\end{equation}
so that
\begin{equation}
I_t=(\Z-(1-\sqrt{\sigma})^2 t+a) c_2^{-1} t^{-1/3}
\end{equation}
with $c_2 = \sigma^{-1/6}(1-\sqrt{\sigma})^{2/3}$, and for a $t$-independent constant $a$ to be specified later. It is also known that~\cite{BF07}
\begin{multline}
K_{t,\sigma,\rm resc}^{\rm stepTASEP}(s_1,s_2):= c_2^{-1}t^{1/3} (1-\sqrt{\sigma})^{X(s_1)-X(s_2)} K_{t,\sigma t}^{\rm stepTASEP}(X(s_1),X(s_2))\\
\to K_{{\cal A}_2}(s_1,s_2)
\end{multline}
as $t\to\infty$ and uniformly for $s_1,s_2$ in bounded sets. Moreover, exponential bounds for the decay of $K_{t,\rm resc}^{\rm stepTASEP}$ ensure that
\begin{equation}
\lim_{t\to\infty}\Pb(x_n(t)\geq X(s))=\det(\Id-K_{{\cal A}_2})_{L^2((s,\infty))} = F_{\rm GUE}(s).
\end{equation}

Now let us focus on the first order correction.
\begin{prop}
Uniformly for $s_1,s_2$ in a bounded set, with the choice $a=1/2$, it holds
\begin{equation}\label{eq59}
K_{t,\sigma,\rm resc}^{\rm stepTASEP}(s_1,s_2)=K_{{\cal A}_2}(s_1,s_2)+c_2^{-1}t^{-1/3}K_{\rm asym}(s_1,s_2)+\Or(t^{-2/3}),
\end{equation}
where $K_{\rm asym}(s_1,s_2)=P(s_1,s_2)-P(s_2,s_1)$, with
\begin{multline}
P(s_1,s_2)= \frac{1}{2} \int_{\R_+} \dx \lambda\, \Ai(s_1+\lambda) \\
\times \left(\Ai'(s_2+\lambda) + s_2 \Ai''(s_2+\lambda) - \frac{1-2 \sqrt{\sigma}}{2 \sqrt{\sigma}}\Ai^{(4)}(s_2+\lambda)\right).
\end{multline}
\end{prop}

\begin{figure}
\begin{center}
\psfrag{S}{$s$}
\psfrag{F}{}
\psfrag{1.0}[cr][][0.8][270]{$1.0$}
\psfrag{0.8}[cr][][0.8][270]{$0.8$}
\psfrag{0.6}[cr][][0.8][270]{$0.6$}
\psfrag{0.4}[cr][][0.8][270]{$0.4$}
\psfrag{0.2}[cr][][0.8][270]{$0.2$}
\psfrag{0.0}[cr][][0.8][270]{$0.0$}
\psfrag{1}[][][0.8]{$1$}
\psfrag{0}[][][0.8]{$0$}
\psfrag{-1}[][][0.8]{$-1\hspace{0.4em}$}
\psfrag{-2}[][][0.8]{$-2\hspace{0.4em}$}
\psfrag{-3}[][][0.8]{$-3\hspace{0.4em}$}
\psfrag{-4}[][][0.8]{$-4\hspace{0.4em}$}
\includegraphics[height=10cm,angle=-90]{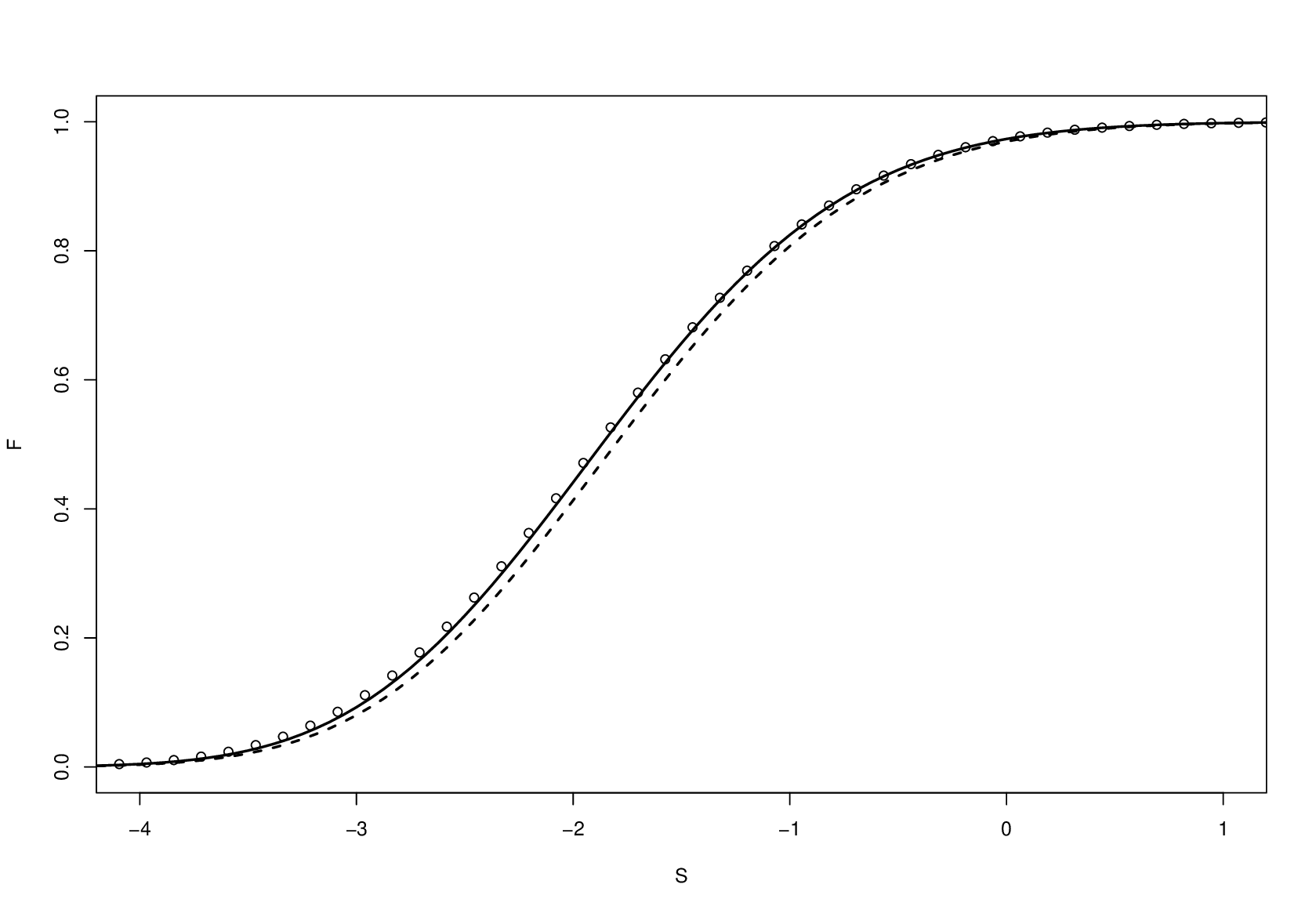}
\caption{Distribution function of $x_{n=[t/4]}(t)$ for TASEP with step initial conditions and $t=1000$. The number of runs is $10^6$. The dots are the plot of \mbox{$(s\in I_t,\Pb(x_{[t/4]}(t)\geq X(s)))$} with $\delta_t=(t/2)^{-1/3}$ and $a=1/2$. The solid line is \mbox{$(s,F_{\rm GUE}(s+\tfrac12\delta_t))$} while the dashed line is \mbox{$(s,F_{\rm GUE}(s+\tfrac12\delta_t-a\delta_t)))$}.}
\label{FigStepCDF}
\end{center}
\end{figure}

\begin{figure}
\begin{center}
\psfrag{S}{$s$}
\psfrag{P}{}
\psfrag{0.4}[cr][][0.8][270]{$0.4$}
\psfrag{0.3}[cr][][0.8][270]{$0.3$}
\psfrag{0.2}[cr][][0.8][270]{$0.2$}
\psfrag{0.1}[cr][][0.8][270]{$0.1$}
\psfrag{0.0}[cr][][0.8][270]{$0.0$}
\psfrag{1}[][][0.8]{$1$}
\psfrag{0}[][][0.8]{$0$}
\psfrag{-1}[][][0.8]{$-1\hspace{0.4em}$}
\psfrag{-2}[][][0.8]{$-2\hspace{0.4em}$}
\psfrag{-3}[][][0.8]{$-3\hspace{0.4em}$}
\psfrag{-4}[][][0.8]{$-4\hspace{0.4em}$}
\includegraphics[height=10cm,angle=-90]{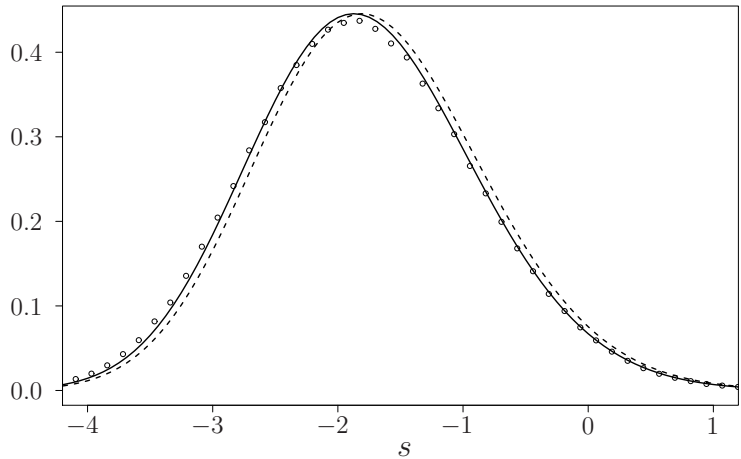}
\caption{Probabilities for $x_{[t/4]}(t)$ for TASEP with step initial conditions and $t=1000$. The number of runs is $10^6$.  The dots are the plot of \mbox{$(s\in I_t,\Pb(x_{[t/4]}(t)= X(s)) \delta_t^{-1})$} with $\delta_t=(t/2)^{-1/3}$ and $a=1/2$. The solid line is \mbox{$(s,F'_2(s))$} while the dashed line is \mbox{$(s,F'_2(s-a\delta_t))$}.}
\label{FigStepDensity}
\end{center}
\end{figure}

\begin{proof}
With $c_1 = 1-2\sqrt{\sigma}$ we can write $X(s)=c_1 t -c_2 st^{1/3}-a$. The rescaled kernel then reads
\begin{multline}
K_{t,\sigma,\rm resc}^{\rm stepTASEP}(s_1,s_2)=\frac{c_2 t^{1/3} (1-\sqrt{\sigma})^{c_2(s_1-s_2)t^{1/3}}}{(2\pi\I)^2}\\
\times \oint_{\Gamma_0}\dx z \oint_{\Gamma_1} \dx w \, \frac{e^{tz}}{e^{tw}} \left(\frac{1-z}{1-w}\right)^{\sigma t} \frac{w^{(\sigma + c_1) t - c_2 s_2 t^{1/3}-a}}{z^{(\sigma+c_1)t-c_2s_1t^{1/3}-a+1}}\frac{1}{z-w}.
\end{multline}
The function $z \mapsto z + \sigma \ln(1-z)-(\sigma+c_1)\ln z$ has a double critical point at $\xi=1-\sqrt{\sigma}$. The steepest descent path for $z$ can be taken such that it comes into $\xi$ with an angle $e^{-2\pi\I/3}$, leaves with an angle $e^{2\pi\I/3}$, and is completed by a piece of a circle around zero of radius strictly larger than $\xi$. The steepest descent path for $w$ comes into $\xi$ with an angle $e^{\pi\I/3}$, leaves it with an angle $e^{-\pi\I/3}$, and is completed by a piece of a circle around zero of radius strictly larger than $1$. By the change of variables
\begin{equation}
z = \xi + c_3^{-1} t^{-1/3} Z, \quad w = \xi + c_3^{-1} t^{-1/3} W
\end{equation}
with $c_3= \sigma^{-1/6}(1-\sqrt{\sigma})^{-1/3}$ and a large $t$ expansion of the integrand, we have
\begin{multline}
K_{t,\sigma,\rm resc}^{\rm stepTASEP}(s_1,s_2)= \frac{1}{(2\pi\I)^2} \int_{\infty e^{-2\pi\I/3}}^{\infty e^{2 \pi\I/3}}\dx Z \int_{\infty e^{-\pi \I/3}}^{\infty e^{\pi \I/3}}\dx W \, \frac{e^{W^3/3 - s_2W}}{e^{Z^3/3-s_1Z}} \frac{1}{W-Z} \\
\times \biggl( 1+  c_2^{-1}t^{-1/3} \biggl((a-1) Z-aW + \frac{s_2W^2-s_1 Z^2}{2} + c_4(Z^4-W^4) \biggr)+\Or(t^{-2/3})\biggr),
\end{multline}
with $c_4=(1-2\sqrt{\sigma})/(4\sqrt{\sigma})$. The choice $a=1/2$ makes the first order correction of the kernel antisymmetric. Finally, we can choose the paths satisfying $\Re Z < \Re W$ and use $\frac{1}{W-Z} = \int_0^\infty \dx\lambda\, e^{-\lambda(W-Z)}$ to obtain (\ref{eq59}).
\end{proof}

\begin{remark}\label{RemFlatvsStep}
Looking at Figures~\ref{FigFlatCDF}--\ref{FigStepDensity} one has the impression that TASEP with alternating initial conditions is already ``closer'' than with step initial conditions to its asymptotics at time $t=1000$, which is confirmed by the data in Table~\ref{Table}. This is to remind the reader that although in both cases the error is $\Or(t^{-2/3})$, depending on the prefactor one still might see some differences of the accuracy for not too large times $t$. The slower convergence for curved vs.\ flat geometry holds also for the PNG model as verified numerically by Richter in his diploma thesis~\cite{Ric11} adapting the numerical approach of Borneman~\cite{Born08}.
\begin{table}
\begin{center}
\begin{tabular}{|p{8em}|l|l|l|l|}
  \hline
  TASEP, $t=1000$ & Mean & Variance & Skewness & Kurtosis \\
  \hline
  & & & & \\[-1em]
  \hline
  Alternating IC & $-0.60495$ & \phantom{$-$}$0.4027$ & \phantom{$-$}$0.282$ & \phantom{$-$}$0.143$ \\
  \hline
  $s\mapsto F_{\rm GOE}(2s)$ & $-0.60327$ & \phantom{$-$}$0.4019(5)$ & \phantom{$-$}$0.293$ & \phantom{$-$}0.165 \\
  \hline
  Relative error & \phantom{$-$}$0.28$ \% & \phantom{$-$}$0.18$ \% & $-3.6$ \% & $-13$ \% \\
  \hline
  & & & & \\[-1em]
  \hline
  Step IC & $-1.7949(7)$ & \phantom{$-$}$0.842(2)$  & \phantom{$-$}$0.19(0)$  & \phantom{$-$}$0.06(7)$ \\
  \hline
  $s\mapsto F_{\rm GUE}(s)$ & $-1.77109$ & \phantom{$-$}$0.8132$ & \phantom{$-$}$0.224$ & \phantom{$-$}$0.094$ \\
  \hline
  Relative error & \phantom{$-$}$1.3$ \% & \phantom{$-$}$3.6$ \% & $-15$ \% & $-29$ \% \\
  \hline
\end{tabular}
\caption{Comparison between alternating and step initial conditions. The data comes from the simulation used for the previous figures.}
\label{Table}
\end{center}
\end{table}

\end{remark}

\section{PASEP}\label{SectPASEP}
Consider the partially asymmetric simple exclusion process on $\Z$ in continuous time with step initial condition. A formula for the one-point distribution of the $n$th particle from the right has been derived in~\cite{TW07,TW08}. The expression is this time not just a Fredholm determinant, but an integral in the complex plane of a Fredholm determinant. A rigorous large time asymptotic analysis is in~\cite{TW08b}, in which it is shown that particles in the rarefaction fan fluctuate asymptotically according to the GUE Tracy-Widom distribution $F_{\rm GUE}$. The scaling limit to be considered is\footnote{There is a minor difference with respect to the papers of Tracy and Widom. To get their framework we need to apply the transformation $x\to -x$.}
\begin{equation}
n=\sigma t\in \Z,\quad X(s)=c_1(\sigma) t-s c_2(\sigma) t^{1/3}-a\in \Z
\end{equation}
where
\begin{equation}
c_1(\sigma)=1-2\sqrt{\sigma},\quad c_2(\sigma)=\sigma^{-1/6} (1-\sqrt{\sigma})^{2/3}.
\end{equation}
The $t$-independent constant $a$ will be specified later. Then in~\cite{TW08b} it is proven that
\begin{equation}
\lim_{t\to\infty} \Pb(x_n(t/\gamma)\geq X(s))=F_{\rm GUE}(s)\quad \textrm{with}\quad \gamma=p-q>0.
\end{equation}
Our result on the first order correction is the following.
\begin{prop} \label{PropPASEP}
Let $p\in (\frac{1}{2},1]$, $q=1-p$, and set
\begin{equation}
a_{p,q}=\sum_{\ell=1}^\infty \frac{q^\ell}{p^\ell-q^\ell}\quad \textrm{and}\quad a=\frac12-\frac{1}{\sqrt{\sigma}}a_{p,q}.
\end{equation}
Then for large time $t$ it holds
\begin{equation}\label{eq33}
\Pb\left(x_n(t/\gamma)\geq X(s)\right)=F_{\rm GUE}(s+\tfrac12 \delta_t)(1+\Or(t^{-2/3})),
\end{equation}
for $s$ in a bounded subset of $I_t=(\Z-c_1(\sigma)t+a)\delta_t$ with $\delta_t=c_2(\sigma)^{-1}t^{-1/3}$.
\end{prop}
\begin{remark}
Note that the previously discussed TASEP with step initial conditions is a special case of this result, with $p=1-q=1$, since $a_{1,0}=0$
\end{remark}
From this result one can easily get the corresponding result for the height function. Let $\eta_x(t)$ be $1$ if there is a particle at site $x$ at time $t$, and zero otherwise. Then the height function is defined by
\begin{equation}
h(x,t)=\left\{
\begin{array}{ll}
2J(0,t)+\sum_{y=0}^{x-1}(1-2\eta_y(t)), & \textrm{ for }x\geq 1, \\
2J(0,t), & \textrm{ for }x=0, \\
2J(0,t)-\sum_{y=x}^{-1} (1-2\eta_y(t)), & \textrm{ for }x\leq -1,
\end{array}
\right.
\end{equation}
where $J(0,t)=\sum_{y\geq 0}\eta_{y}(t)$. To get the result for the $h$ from Proposition~\ref{PropPASEP} one simply uses the identity
\begin{equation}\label{eqHX}
\Pb(x_{j}(t)\geq x)=\Pb(h(x,t)\geq 2j+x)
\end{equation}
with the following result.
\begin{cor}\label{CorShiftPASEP}
Let $x=\xi t\in\Z$ and set
\begin{equation}\label{eq4.8}
H(s)=\frac12(1+\xi^2)t-s \frac{(1-\xi^2)^{2/3}}{2^{1/3}}t^{1/3}+\tilde a,\quad \tilde a=2 a_{p,q}-1.
\end{equation}
Then, for large $t$ it holds
\begin{equation}\label{eq4.9}
\Pb(h(x,t/\gamma)\geq H(s)) = F_{\rm GUE}(s+\tfrac12\delta_t^h)(1+\Or(t^{-2/3})),
\end{equation}
for $s\in I^h_t$, where $I^h_t$ is defined by the requirement $H(s)\in 2\Z$ if $x$ is even and $H(s)\in 2\Z+1$ if $x$ is odd. Since the lattice width of $h$ is $2$, we have $\delta_t^h= 2^{4/3}(1-\xi^2)^{-2/3}t^{-1/3}$.\\
From this result it follows that the critical value $p_c$ of the asymmetry in PASEP such that the density (and the moments) of the rescaled integrated current are correct up to order $\Or(t^{-2/3})$ is the solution of
\begin{equation}
a_{p_c,1-p_c}=\frac12\quad \iff \quad p_c=0.78227\,87862{\ldots}
\end{equation}
\end{cor}
In Figure~\ref{FigPasep} we plot the function $2 a_{p,1-p}-1$.
\begin{proof}[Proof of Corollary~\ref{CorShiftPASEP}]
Let us define a linearization of the distribution functions by
\begin{equation}
\widetilde F_t(s):=
\begin{cases}
\Pb(x_n(t/\gamma)\geq c_1(\sigma) t-s c_2(\sigma) t^{1/3}-\frac12+\frac{1}{\sqrt{\sigma}}a_{p,q}), & \textrm{if }s\in I_t,\\
\textrm{linear interpolation}, & \textrm{otherwise,}
\end{cases}
\end{equation}
and similarly
\begin{equation}
\widetilde F^h_t(s):=
\begin{cases}
\Pb(h(\xi t,t/\gamma)\geq H(s), & \textrm{if }s\in I^h_t,\\
\textrm{linear interpolation}, & \textrm{otherwise.}
\end{cases}
\end{equation}
Then, Proposition~\ref{PropPASEP} tell us that
\begin{equation}\label{eq4.11}
\widetilde F_t(s-\tfrac12\delta_t)= F_{\rm GUE}(s)(1+\Or(t^{-2/3})),\quad s\in\R,
\end{equation}
and we want to show that
\begin{equation}
\widetilde F^h_t(s-\tfrac12\delta_t^h)= F_{\rm GUE}(s)(1+\Or(t^{-2/3})),\quad \textrm{for }s\in I_t^h+\tfrac12 \delta_t^h.
\end{equation}
For $s\in I_t^h+\tfrac12 \delta_t^h$, using (\ref{eqHX}) we have
\begin{equation}\label{eq4.10}
\begin{aligned}
\widetilde F_t^h(s-\tfrac12\delta_t^h)& =\Pb(h(\xi t,t/\gamma)\geq H(s-\tfrac12 \delta_t^h))\\
&=\Pb(h(\xi t,t/\gamma)\geq H(s)+1)= \Pb(x_{\sigma t}(t/\gamma)\geq \xi t)
\end{aligned}
\end{equation}
with $\sigma=(H(s)+1-\xi t)/(2t)$. With this value of $\sigma$, an algebraic computation gives
\begin{equation}
\xi t=c_1(\sigma) t-s c_2(\sigma) t^{1/3}-(a-1/2)+\Or(t^{-1/3}).
\end{equation}
Since $s c_2(\sigma) t^{1/3}+(a-1/2)=(s-\tfrac12 \delta_t) c_2(\sigma) t^{1/3}+a$ we get that
\begin{equation}
\widetilde F_t^h(s-\tfrac12\delta_t^h)= (\ref{eq4.10}) =\widetilde F_t(s-\tfrac12 \delta_t+\Or(t^{-2/3}))=F_{\rm GUE}(s)(1+\Or(t^{-2/3}))
\end{equation}
where in the last step we used (\ref{eq4.11}) coming from Proposition~\ref{PropPASEP}.
\end{proof}
\begin{remark}\label{remWASEP}
As $p\to 1/2$ the model becomes close to the WASEP studied in~\cite{SS10,SS10a,SS10b,ACQ10}. In particular, our result matches the limit behavior of~\cite{SS10b}. Indeed, when the asymmetry $\beta:=2p-1\to 0$, the shift for the fitting of the density behaves as
\begin{equation}
a_{p,1-p} \simeq \frac{\gamma_{\rm E}-\ln(2\beta)}{2\beta}+\frac14+\Or(\beta),
\end{equation}
with $\gamma_{\rm E}=-\partial_x\ln(\Gamma(x))\big|_{x=1}=-0.57721\,56649\ldots$ the Euler constant, so that for the height function the shift is then $\tilde a=\beta^{-1}(\ln(2\beta)-\gamma_{\rm E})-\tfrac12+\Or(\beta)$.
\end{remark}
\begin{figure}
\begin{center}
\psfrag{p}[b]{$p$}
\psfrag{a}{$2 a_{p,1-p}-1$}
\psfrag{-1}[r][][0.8]{$-1$}
\psfrag{0}[r][][0.8]{$0$}
\psfrag{1}[r][][0.8]{$1$}
\psfrag{2}[r][][0.8]{$2$}
\psfrag{3}[r][][0.8]{$3$}
\psfrag{0.5}[t][][0.8]{$0.5$}
\psfrag{0.6}[t][][0.8]{$0.6$}
\psfrag{0.7}[t][][0.8]{$0.7$}
\psfrag{0.8}[t][][0.8]{$0.8$}
\psfrag{0.9}[t][][0.8]{$0.9$}
\psfrag{1.0}[t][][0.8]{$1.0$}
\includegraphics[height=6cm]{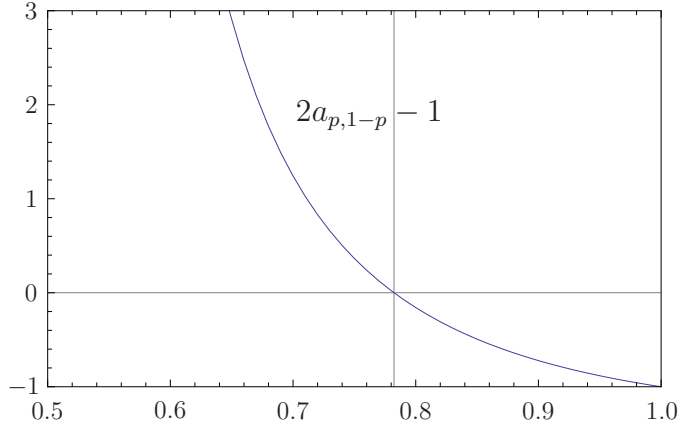}
\caption{The function $p\mapsto 2 a_{p,1-p}-1$ for $p\in (1/2,1]$. It crosses the ordinate axis as $p=p_c=0.78227\,87862{\ldots}$.}
\label{FigPasep}
\end{center}
\end{figure}

\begin{proof}[Proof of Proposition~\ref{PropPASEP}]
As for PNG and TASEP, we indicate the main steps of the asymptotic analysis to get (\ref{eq33}) for $s$ in a bounded set, but we will derive bounds for $|s|\to\infty$ needed to determine moment convergence.
Set $u=c_1 t - c_2 s t^{1/3}- a$ and $\tau=q/p<1$. As shown in~\cite{TW09},
\begin{equation}
\Pb\bigl(x_n(t/\gamma)\geq u\bigr) = \frac{1}{2\pi\I} \oint \frac{\dx\mu}{\mu}\, (\mu;\tau)_\infty \det\bigl(\Id+\mu J_\mu),
\end{equation}
where $(\mu;\tau)_\infty$ is the $q$-Pochhammer symbol (see Appendix~\ref{SectPochhammer} for identities) and the integral is taken over a circle around the origin with radius in the interval $(0,\tau)$. The operator $J_\mu$ has kernel
\begin{equation}
J_\mu(\eta,\eta') = \frac{1}{2\pi\I} \oint \dx \zeta\,\frac{\exp \bigl(\frac{t \zeta}{1-\zeta}\bigr) (1-\zeta)^u \zeta^n}{\exp\bigl(\frac{t \eta'}{1-\eta'}\bigr) (1-\eta')^u (\eta')^{n+1}}\, \frac{f\bigl(\mu,\frac{\zeta}{\eta'}\bigr)}{\zeta-\eta},
\end{equation}
where $\eta,\eta'$ are on a circle around $0$ with radius $r\in(\tau,1)$ and $\zeta$ runs on a circle around $0$ with radius in $(1,r/\tau)$. For $1<\lvert z \rvert < \tau^{-1}$, the function $f$ is given by
\begin{equation}
f(\mu,z)=\sum_{k=-\infty}^\infty \frac{\tau^k}{1-\tau^k \mu} \, z^k,
\end{equation}
which extends analytically to $\C^\ast \setminus \{\tau^k:k\in\Z\}$.

The function $\zeta \mapsto \frac{\zeta}{1-\zeta} + \sigma \ln \zeta + c_1 \ln(1-\zeta)$ has a double critical point at $\xi=-\frac{\sqrt{\sigma}}{1-\sqrt{\sigma}}$, and the steepest descent path can be taken such that the $\eta$-contour for $J_\mu$ is a pair of rays from $\xi$ in the directions $\pm \pi/3$ completed by a circle around zero of radius strictly smaller than $1$, and the $\zeta$-contour is a pair of rays from $\xi-t^{-1/3}$ in the directions $\pm 2\pi/3$ completed by a circle around zero of radius strictly larger than $1$. We then do the transformations
\begin{equation}
\zeta = \xi +  c_3^{-1} t^{-1/3} z, \quad \eta = \xi + c_3^{-1} t^{-1/3}w , \quad \eta'=\xi + c_3^{-1} t^{-1/3} \tilde w,
\end{equation}
with $c_3=\sigma^{-1/6}(1-\sqrt{\sigma})^{5/3}$. Expanding $f$ around $z=1$ yields
\begin{equation}
\mu\,f(\mu,z)= \frac{1}{1-z} + g(\mu) + \Or\bigl(\lvert z-1\rvert\bigr),
\end{equation}
where
\begin{equation}
g(\mu)=\sum_{k=0}^\infty \frac{\mu \tau^k}{1-\mu \tau^k}+\sum_{k=1}^\infty \frac{\tau^k}{\tau^k-\mu}.
\end{equation}
This expansion is obtained by dividing the series into $\{0,1,\ldots\}$ and $\{\ldots,-2,-1\}$, using
\begin{equation}
\frac{\mu \tau^k}{1-\tau^k \mu}=\frac{1}{1-\tau^k \mu}-1,
\end{equation}
and then change the variable $k\to -k$ in one of the sum.
After a large $t$ expansion, the kernel $\mu J_\mu$ can be written as
\begin{multline} \label{eq:2}
\frac{1}{2\pi\I}\int_{e^{-2\pi \I/3}\infty}^{e^{2 \pi\I/3}\infty} \dx z\, \frac{e^{\tilde w^3/3-\tilde w(s+\delta_t/2)}}{e^{z^3/3-z (s+\delta_t/2)}}\frac{e^{(\tilde w-z)\delta_t/2}}{(z-w)(\tilde w-z)} \\
\times \frac{\exp\bigl(\bigl(c_4 \tilde w^4 - \frac{1}{2} s\tilde w^2 - \bigl(\frac{g(\mu)}{\sqrt{\sigma}}+a\bigr) \tilde w \bigr)c_2^{-1}t^{-1/3}+\Or(t^{-2/3})\bigr)}{\exp\bigl(\bigl(c_4 z^4 - \frac{1}{2} sz^2- \bigl(\frac{g(\mu)}{\sqrt{\sigma}}+a\bigr)z \bigr)c_2^{-1}t^{-1/3}+\Or(t^{-2/3})\bigr)}
\end{multline}
with $c_4=\frac{1+2\sqrt{\sigma}}{4\sqrt{\sigma}}$.
Since $\Re (z-\tilde w)<0$, we have
\begin{equation} \label{eq:3}
\frac{e^{(z-\tilde w)(s+\delta_t/2)}}{\tilde w-z}= \int_{s+\delta_t/2}^\infty \dx x\,e^{(z-\tilde w)x},
\end{equation}
and plugging (\ref{eq:3}) into (\ref{eq:2}) gives
\begin{multline}
\frac{1}{2\pi\I}\int_{s+\delta_t/2}^\infty \dx x \int_{e^{-2\pi \I/3}\infty}^{e^{2 \pi\I/3}\infty} \dx z \, \frac{e^{\tilde w^3/3-\tilde wx}}{e^{z^3/3-z x}} \frac{1}{z-w}  \\
\times \frac{\exp\bigl(\bigl(c_4 \tilde w^4 - \frac{1}{2} s\tilde w^2 - \bigl(\frac{g(\mu)}{\sqrt{\sigma}}-\frac12+a\bigr) \tilde w \bigr)c_2^{-1}t^{-1/3}+\Or(t^{-2/3})\bigr)}{\exp\bigl(\bigl(c_4 z^4 - \frac{1}{2} sz^2- \bigl(\frac{g(\mu)}{\sqrt{\sigma}}-\frac12+a\bigr)z \bigr)c_2^{-1}t^{-1/3}+\Or(t^{-2/3})\bigr)}
\end{multline}
The operator $J_\mu$ is the product $C_1\,C_2\,C_3$, where the factors have kernels
\begin{equation}
 \begin{aligned}
  C_1(w,z)&= \frac{1}{2\pi \I}\,\frac{e^{-z^3/3}}{z-w},\\
C_2(z,x)&= \frac{1}{2\pi \I}\, e^{xz} \, e^{- c_2^{-1}t^{-1/3}\big[c_4 z^4 - sz^2 /2 -z(g(\mu)/\sqrt{\sigma}-1/2+a)\big] +\Or(t^{-2/3})}, \\
C_3(x,\tilde w)&= e^{\tilde w^3/3-\tilde wx}\, e^{c_2^{-1}t^{-1/3}\big[c_4 \tilde w^4-s\tilde w^2/2- \tilde w (g(\mu)/\sqrt{\sigma}-1/2+a)\big] +\Or(t^{-2/3})}.
 \end{aligned}
\end{equation}
The operator $C_3\,C_1\,C_2$, which has the same Fredholm determinant, acts on $L^2(s+\delta_t/2,\infty)$ and has kernel with $(x,y)$ entry given by
\begin{multline} \label{eq:4}
\frac{1}{(2\pi\I)^2} \int_{e^{-2\pi \I/3}\infty}^{e^{2 \pi\I/3}\infty} \dx z  \int_{e^{-\pi \I/3}\infty}^{e^{\pi \I/3}\infty}\dx w\, \frac{e^{w^3/3- wx}}{e^{z^3/3-zy}} \frac{1}{z - w} \Bigl( 1  +\Or(t^{-2/3})  \\
+\bigl( \bigl(c_4 (w^4-z^4)- \tfrac{1}{2} s(w^2-z^2)- \bigl(\tfrac{g(\mu)}{\sqrt{\sigma}}-\tfrac12+a\bigr) (w-z) \bigr)c_2^{-1}t^{-1/3}\Bigr).
\end{multline}
Using again $\Re(z-w)<0$, we can write $\frac{1}{w-z}=\int_0^\infty \dx \lambda e^{-\lambda(w-z)}$. Thus, (\ref{eq:4}) equals
\begin{equation}
- K_{{\cal A}_2}(x,y) + c_2^{-1}t^{-1/3} \bigl(K_\text{asym}(x,y)+K_\text{sym}(x,y)\bigr) + \Or(t^{-2/3}),
\end{equation}
where $K_\text{asym}(x,y)=P(x,y)-P(y,x)$ with
\begin{equation}
P(x,y)= \int_0^\infty \dx \lambda \Ai(x+\lambda) \left[c_4 \frac{\dx^4}{\dx y^4}-\frac{s}{2} \frac{\dx^2}{\dx y^2}\right]\!\Ai(y+\lambda)
\end{equation}
is asymmetric, and $K_\text{sym}(x,y)=-\bigl(\frac{g(\mu)}{\sqrt{\sigma}}-\frac12+a\bigr)\Ai(x)\Ai(y)$ is symmetric.
Hence, we have
\begin{multline}\label{eq4.18}
\Pb\bigl(x_m(t/\gamma)\geq u\bigr)=F_{\rm GUE}(s+\delta_t/2)\\
\times\Bigl(1 - \bigl(\tfrac{G}{\sqrt{\sigma}}-\tfrac12+a\bigr) \Tr\bigl( (\Id - \chi_s K_{{\cal A}_2}\chi_s)^{-1} \chi_s (\Ai \otimes \Ai) \chi_s\bigr)t^{-1/3} \\ + \Or(t^{-2/3})\Bigr).
\end{multline}
with $\chi_s$ the projection onto $(s+\delta_t/2,\infty)$ and
\begin{equation}
G = \frac{1}{2\pi\I} \oint_{\tau < \lvert \mu \rvert < 1} \frac{\dx\mu}{\mu}  g(\mu) (\mu;\tau)_\infty.
\end{equation}
We will show that $G=a_{p,q}$ so that by choosing $a=\frac12-\frac{1}{\sqrt{\sigma}}a_{p,q}$ the prefactor of the first order correction vanishes. First note that
\begin{equation}
\frac{1}{2\pi \I} \oint_{\tau<\lvert \mu \rvert < 1} \frac{\dx\mu}{\mu} (\mu;\tau)_\infty \sum_{k=0}^\infty \frac{\mu\tau^k}{1-\tau^k \mu} =0,
\end{equation}
as the integrand has no poles inside the unit circle. So, we have
\begin{equation}
G=\frac{1}{2\pi\I} \oint_{\tau<\lvert \mu \rvert < 1} \frac{\dx \mu}{\mu} (\mu;\tau)_\infty \sum_{k=1}^\infty \frac{\tau^k}{\tau^k-\mu}.
\end{equation}
We use
\begin{equation}
\frac{1}{\mu}\frac{\tau^k}{\tau^k-\mu} = \frac{1}{\mu}+\frac{1}{\tau^k-\mu},
\end{equation}
the fact that $(\mu;\tau)_\infty$ is analytic inside the integration domain, so that the sum of the contributions of the simple poles gives
\begin{equation}
G=\sum_{k=1}^\infty \left(1-(\tau^k;\tau)_\infty\right).
\end{equation}
There is a simpler expression for $G$. Using the identity (\ref{eq102}) we get
\begin{equation}
G=-\sum_{k=1}^\infty\sum_{\ell=1}^\infty \frac{(-1)^\ell \tau^{\ell(\ell-1)/2}}{(\tau;\tau)_\ell}\,\tau^{k \ell}
=\sum_{\ell=1}^\infty \frac{(-1)^\ell \tau^{\ell(\ell-1)/2}}{(\tau;\tau)_\ell}\,\frac{\tau^\ell}{\tau^\ell-1}
\end{equation}
Then we use (\ref{eqQHypDef}) and
$(\alpha;\tau)_\ell /(\alpha\tau;\tau)_\ell =(\alpha-1)/(\alpha \tau^\ell-1)$
to get
\begin{equation}
\begin{aligned}
G&=\lim_{\alpha\to 1} \frac{1}{\alpha-1}\sum_{\ell=1}^\infty \frac{(-1)^\ell \tau^{\ell(\ell-1)/2}(\alpha;\tau)_\ell}{(\tau;\tau)_\ell(\alpha\tau;\tau)_\ell}\tau^\ell\\
&=\lim_{\alpha\to 1} \frac{1}{\alpha-1} \biggl( \QHyp{1}{1}{\alpha}{\alpha\tau}{\tau}{\tau}-1\biggr)\\
&=\lim_{\alpha\to 1} \frac{1}{\alpha-1}\lim_{\beta\to\infty}\biggl( \QHyp{2}{1}{\alpha,\beta}{\alpha\tau}{\tau}{\frac{\tau}{\beta}}-1 \biggr)
\end{aligned}
\end{equation}
where we used (\ref{eqA4}) in the last equality. Finally, the $q$-Gauss identity (\ref{eqQGauss}) leads to
\begin{equation}
\begin{aligned}
G&=\lim_{\alpha\to 1} \frac{1}{\alpha-1}\lim_{\beta\to\infty}\bigg(\frac{(\alpha \tau/\beta;\tau)_\infty (\tau;\tau)_\infty}{(\alpha \tau;\tau)_\infty (\tau/\beta;\tau)_\infty} -1\bigg) =
-\frac{\partial_{\alpha}(\alpha\tau;\tau)_\infty}{(\tau;\tau)_\infty}\bigg|_{\alpha=1}\\
& =- \partial_\alpha \ln[(\alpha\tau;\tau)_\infty]\big|_{\alpha=1} = -\partial_\alpha \sum_{\ell=0}^\infty \ln(1-\alpha \tau^{\ell+1})\bigg|_{\alpha=1} = \sum_{\ell=1}^\infty\frac{\tau^\ell}{1-\tau^\ell}.
\end{aligned}
\end{equation}
Replacing $\tau=q/p$ leads to $G=a_{p,q}$.
This and (\ref{eq4.18}) shows that
\begin{equation}
\Pb\bigl(x_n(t/\gamma)\geq u\bigr)=F_{\rm GUE}(s)(1+\Or(t^{-2/3})).
\end{equation}
\end{proof}

\appendix

\section{Discrete sums versus integrals}
\begin{lem}\label{LemmaSumIntNdim}
Let $f:\R^n\to \R$ be a smooth function with $\partial_j\partial_k f\in L^1(\R^n)$, for all $j,k=1,\ldots,n$. Then, for $\delta>0$ (small)
\begin{equation}
\bigg|\delta^n \sum_{x\in(\Z_+^*)^n} f(x\delta) - \int_{(-\frac\delta2,\infty)^n}\dx^nx\, f(x)\bigg|=\Or(\delta^2) \sum_{j,k=1}^n \int_{\R_+^n}\dx^nx\, |\partial_j\partial_k f(x)|.
\end{equation}
\end{lem}
\begin{proof}
We first rewrite
\begin{equation}
\int_{(-\frac\delta2,\infty)^n}\dx^nx\, f(x)=\sum_{x\in(\Z_+^*)^n}\int_{[-\frac\delta2,\frac\delta2]^n}  \dx^ny\, f(x\delta+y)
\end{equation}
and use Taylor development
\begin{equation}
f(x\delta+y)= f(x\delta)+\sum_{j=1}^n \partial_j f(x\delta) y_j + \frac12 \sum_{j,k=1}^n y_j y_k R_{j,k}(\delta,y),
\end{equation}
with
\begin{equation}
|R_{j,k}(\delta,y)|\leq \max_{u\in [-\frac\delta2,\frac\delta2]^2} |\partial_j\partial_k f(x\delta+u)|,\quad\textrm{for }y\in [-\tfrac\delta2,\tfrac\delta2]^n
\end{equation}
to obtain (after integrating over $y$),
\begin{equation}
\bigg|\int_{(-\frac\delta2,\infty)^n}\hspace{-2em} \dx^nx\, f(x)-\sum_{x\in(\Z_+^*)^n} \delta^n f(x\delta)\bigg|
\leq \frac{\delta^2}{12}\sum_{x\in(\Z_+^\ast)^n}\sum_{j,k=1}^n \delta^n \max_{u\in [-\frac\delta2,\frac\delta2]^2} |\partial_j\partial_k f(x\delta+u)|.
\end{equation}
The statement then follows because the sum over $x$ converges, as $\delta\to 0$, to $\sum_{j,k=1}^n \int_{\R_+^n}\dx^nx\, |\partial_j\partial_k f(x)|$.
\end{proof}

\begin{lem}\label{lemBoundKgeneral}
Let $f(x_1,\ldots,x_n)=\det (K(x_i,x_j))_{1\le i,j\le n}$ with the kernel $K$ satisfying
\begin{equation} \label{eq5}
\max\{ \lvert K(x_1,x_2) \rvert , \lvert \partial_i K(x_1,x_2) \rvert,\lvert \partial_i \partial_j K(x_1,x_2)\rvert\} \le C e^{-c(x_1+x_2)}
\end{equation}
for all $x_1,x_2 \in (s,\infty)$, $i,j \in \{1,2\}$ and some positive constants $c$, $C$. Then,
\begin{equation}
\lvert \partial_i \partial_j f(x_1,\ldots,x_n)\rvert \le 4\, C^n n^{n/2} \prod_{k=1}^n e^{-2cx_k}
\end{equation}
for all $1\le i,j \le n$.
\end{lem}

\begin{proof}
Let $K_{i,}$ (resp.\ $K_{,j}$) be the matrix $(K(x_i,x_j))_{1\le i,j\le n}$ with the $i$th row (resp.\ the $j$th column) replaced by its derivative w.r.t.\ the first (resp.\ the second) variable. Then,
\begin{equation}
\partial_i f(x_1,\ldots,x_n)= \det K_{i,}+\det K_{,i}
\end{equation}
and from this
\begin{equation}
\partial_i \partial_j f(x_1,\ldots,x_n) = \det K_{ij,}+\det K_{i,j}+\det K_{j,i} + \det K_{,ji}
\end{equation}
with $K_{ij,}=(K_{i,})_{j,}$. By Hadamard's bound, the absolute value of an $n\times n$ determinant with entries in the closed unit disk is bounded by $n^{n/2}$. It then follows
\begin{equation}
\lvert \partial_i \partial_j f(x_1,\ldots,x_n)\rvert \le 4\,C^n n^{n/2} \prod_{k=1}^n e^{-2c x_k}
\end{equation}
for any $1\le i,j\le n$.
\end{proof}

\begin{lem}\label{LemAiryKernels}
For the Airy kernels, $K_{{\cal A}_2}(x,y)=\int_{\R_+}\dx \lambda \, \Ai(x+\lambda)\Ai(y+\lambda)$ and $K_{{\cal A}_1}(x,y)=\Ai(x+y)$, assumption (\ref{eq5}) of Lemma~\ref{lemBoundKgeneral} is satisfied.
\end{lem}
\begin{proof}
It is easy to see from the integral representation
\begin{equation}\label{eqAiry}
\Ai(x)=\frac{1}{2\pi\I} \int_{\infty e^{-\pi \I/3}}^{\infty e^{\pi \I/3}} \dx z\, e^{z^3/3-z\, x}=\frac{1}{2\pi\I} \int_{\I\R+\delta} \dx z\, e^{z^3/3-z\, x},\quad \e>0,
\end{equation}
that for any $\delta>0$, there exists a constant $C_\delta \in (0,\infty)$ so that
\begin{equation}\label{eqBoundsAirys}
\max \{ \lvert \Ai(x)\rvert, \lvert \Ai'(x)\rvert, \lvert \Ai''(x) \rvert \} \le C_\delta e^{-\delta x}
\end{equation}
uniformly in $x\in\R$.

In the case $K(x_1,x_2)=\int_{\R_+}\dx \lambda \, \Ai(x_1+\lambda)\Ai(x_2+\lambda)$ we get from the bounds (\ref{eqBoundsAirys}) with $\e=\frac{1}{2}$, after integration with respect to $\lambda$, that
\begin{equation}\label{eq99}
\max \{ \lvert K(x_1,x_2)\rvert, \lvert \partial_i K(x_1,x_2) \rvert, \lvert \partial_i \partial_j K(x_1,x_2)\rvert \} \leq C e^{-(x_1+x_2)/2}
\end{equation}
for some constant $C>0$ and all $i,j\in \{1,2\}$.

The case $K(x_1,x_2)=\Ai(x_1+x_2)$ is even easier, since the bound (\ref{eq99}) comes directly from (\ref{eqBoundsAirys}).
\end{proof}

\section{\texorpdfstring{$q$-Pochhammer symbols, $q$-hypergeometric functions}{q-Pochhammer symbols, q-hypergeometric functions}}\label{SectPochhammer}
Here we collect some identities on $q$-Pochhammer symbols and $q$-hyper\-geometric functions, used for the PASEP. We use the standards as in~\cite{KS96}. The $q$-Pochhammer symbol is defined by
\begin{equation}
(\mu;q)_\infty=\prod_{k=0}^\infty (1-\mu q^k),\quad\textrm{and}\quad
(\mu;q)_n=\prod_{k=0}^{n-1} (1-\mu q^k).
\end{equation}
They satisfies the following identities:
\begin{equation}\label{eq102}
(\mu;q)_n =\frac{(\mu;q)_\infty}{(\mu q^n;q)_\infty},\quad
(\mu;q)_\infty =\sum_{n=0}^\infty \frac{(-1)^n q^{n(n-1)/2}}{(q;q)_n}\mu^n
\end{equation}
so that in particular $(0;q)_\infty=1$ and $(1;q)_\infty=0$.

The $q$-hypergeometric function is defined by
\begin{equation}\label{eqQHypDef}
\QHyp{r}{s}{a_1,\ldots,a_r}{b_1,\ldots,b_s}{q}{z}=\sum_{n=0}^\infty
\frac{(a_1;q)_n\cdots (a_r;q)_n}{(b_1;q)_n\cdots (b_s;q)_n}\frac{z^n}{(q;q)_n}\left((-1)^n q^{n(n-2)/2}\right)^{1+s-r}.
\end{equation}

In particular, it holds
\begin{equation}\label{eqA4}
\QHyp{r-1}{s}{a_1,\ldots,a_{r-1}}{b_1,\ldots,b_s}{q}{z}
=\lim_{a_r\to\infty}
\QHyp{r}{s}{a_1,\ldots,a_r}{b_1,\ldots,b_s}{q}{\frac{z}{a_r}}.
\end{equation}
The $q$-Gauss identity is
\begin{equation}\label{eqQGauss}
\QHyp{2}{1}{\alpha,\beta}{\gamma}{q}{\gamma}
=\frac{(\gamma/\alpha;q)_\infty (\gamma/\beta;q)_\infty}{(\gamma;q)_\infty (\gamma/(\alpha\beta);q)_\infty}.
\end{equation}


\end{document}